\documentclass[a4paper,twocolumn,11pt,accepted=2024-10-19]{quantumarticle}
\pdfoutput=1
\usepackage[utf8]{inputenc}
\usepackage[english]{babel}
\usepackage[T1]{fontenc}
\usepackage{amsmath}
\usepackage{hyperref}
\usepackage{lipsum}
\usepackage{txfonts}
\usepackage[normalem]{ulem}

\usepackage{amsthm,amssymb,amsxtra}
\usepackage{mathtools}

\usepackage[dvipsnames]{xcolor}

\usepackage{tikz}
\usetikzlibrary{quantikz}
\usepackage{braket}
\usepackage{kantlipsum}

\usepackage[export]{adjustbox}

\usepackage{float}
\makeatletter
\let\newfloat\newfloat@ltx
\makeatother
\usepackage{algorithm}
\usepackage{algpseudocode}

\usepackage{color}
\usepackage{graphicx}
\usepackage{comment}
\usepackage{enumitem}

\usepackage{url}

\usepackage{fancyref}
\usepackage{hyperref}

\DeclareMathOperator*{\argmax}{arg\,max}

\newtheorem{theorem}{Theorem}[section]
\newtheorem{lemma}[theorem]{Lemma}

\newtheorem{corollary}[theorem]{Corollary}

\newcommand{\lrb}[1]{\left ( #1 \right )}

\renewcommand{\exp}[1]{\operatorname{exp} \lrb{#1}}

\renewcommand{\log}[1]{\operatorname{log} \lrb{#1}}

\algrenewcommand\algorithmicrequire{\textbf{Input:}}
\algrenewcommand\algorithmicensure{\textbf{Output:}}

\begin{document}

\title{On proving the robustness of algorithms for early fault-tolerant quantum computers}

\author{Rutuja Kshirsagar}
\affiliation{Zapata Computing Inc., 100 Federal Street, Boston, MA 02110, USA}
\thanks{The author has moved to Fujitsu Research of America, Inc.}
\orcid{0000-0001-6210-0259}
\email{rkshirsagar@fujitsu.com}
\author{Amara Katabarwa}
\affiliation{Zapata Computing Inc., 100 Federal Street, Boston, MA 02110, USA}
\orcid{0000-0002-0078-5202}
\email{amara@zapata.ai}
\author{Peter D.\ Johnson}
\affiliation{Zapata Computing Inc., 100 Federal Street, Boston, MA 02110, USA}
\email{peter@zapatacomputing.com}

\maketitle

\begin{abstract}
  The hope of the quantum computing field is that quantum architectures are able to scale up and realize fault-tolerant quantum computing.
Due to engineering challenges, such ``cheap'' error correction may be decades away.
In the meantime, we anticipate an era of ``costly'' error correction, or \emph{early fault-tolerant quantum computing}.
Costly error correction might warrant settling for error-prone quantum computations.
This motivates the development of quantum algorithms which are robust to some degree of error
as well as methods to analyze their performance in the presence of error.
Several such algorithms have recently been developed; what is missing is a methodology to analyze their robustness. To this end,
we introduce a randomized algorithm for the task of phase estimation and give an analysis of its performance under two simple noise models.
In both cases the analysis leads to a noise threshold, below which arbitrarily high accuracy can be achieved by increasing the number of samples used in the algorithm.
As an application of this general analysis, we compute the maximum ratio of the largest circuit depth and the dephasing scale such that performance guarantees hold.
We calculate that the randomized algorithm can succeed with arbitrarily high probability as long as the required circuit depth is less than $0.916$ times the dephasing scale. 
\end{abstract}

Note: Equation (4) and output of Algorithm1 have been updated to accommodate for a missed negative sign. This does not affect any analysis presented in this paper.

\section{Introduction}
\label{sec:intro}

For several decades researchers have been in pursuit of quantum algorithms that would solve important problems in areas such as chemical and materials simulation~\cite{aspuru2005simulated}, optimization~\cite{durr1996quantum}, and cryptography~\cite{shor1994algorithms}.
It has been assumed that the machine needed to solve such problems is a large-scale fault-tolerant quantum computer~\cite{jones2012layered}.
Recent work has sought to develop methods for solving these problems using smaller, error-prone quantum devices~\cite{peruzzo2014variational, farhi2014quantum, anschuetz2019variational}.
Unfortunately, for many application areas, there is mounting evidence that some degree of quantum error correction will be needed to achieve quantum advantage in those areas~\cite{gonthier2020identifying, johnson2022reducing}.

Quantum error correction is used to reduce error rates of quantum operations at the cost of overhead in number of physical qubits and operating time~\cite{devitt2013quantum}.
In the long-term the quantum computing field hopes to realize large-scale fault-tolerant quantum architectures that are able to implement (nearly) error-free, large quantum computations with relatively little quantum error correction overhead~\cite{fowler2012surface, bombin2021interleaving}.
It is possible, however, that for the early quantum computers that might be used to achieve quantum advantage, error reduction with quantum error correction will require large overhead~\cite{fellous2021limitations}
due to the challenges of scaling up system sizes.
For instance, although the engineering achievements are impressive, recent advances in experimental demonstrations of quantum error correction~\cite{acharya2022suppressing}
realized a only 4\% reduction in logical state initialization error rates as the overhead is increased from a distance-3 to a distance-5 surface code.
This is in contrast to predictions~\cite{kim2022fault} of realizing a 90\% to 96\% reduction in error rate for the same increase in surface code distance in the future. 

In the era of early fault-tolerant quantum computing (EFTQC), when error reduction is costly, it will be important to use error correction resources prudently.
If a quantum algorithm is able to tolerate a degree of error, then this tolerance should be exploited to reduce the error correction resources consumed in the computation.
Quantum advantage will be realized when the quantum resources needed to solve useful computational problems can be supplied by quantum hardware.
Over time there will be an increase in the quantum resources that can be supplied.
But, considering the desire to realize quantum advantage as soon as possible, it is pertinent to determine the minimal quantum resources needed to run a given quantum algorithm.
Considering that these quantum resources are used to reduce error, \emph{what is the minimal error reduction required to ensure that a given algorithm will succeed?}
To answer this question, we must investigate how error in quantum computations impacts the performance of quantum algorithms.

Examples of previous work on analyzing the robustness of quantum algorithms
have been empirical~\cite{devitt2004robustness}, have focused on asymptotic scaling statements~\cite{regev2008impossibility},
or have considered error that is disconnected from a physical model~\cite{gilyen2018quantum}.
More recently, there has been 
increasing focus~\cite{wang2021minimizing, tanaka2021amplitude, giurgica2022low} on the development of quantum algorithms that have a built-in robustness.
Through such work, empirical performance models have been developed~\cite{johnson2022reducing} and tested~\cite{katabarwa2021reducing}.
However, for the purposes of large-scale resource estimates, which help to predict the onset of quantum advantage with early fault-tolerant quantum computers, there is a need for provable robustness results which are more closely connected with realistic error models. 

In this work, we develop an algorithm for the archetypal task of \emph{phase estimation} (see Section \ref{sec:alg_intro}). 
The motivating feature of our algorithm is the fact that it facilitates a simple analysis of its robustness under various algorithmic noise models (see Section \ref{sec:robustness}). While a robustness analysis can be carried out for other quantum phase estimation algorithms, a goal of this work is to develop a foundation for the robustness analysis of algorithms beyond quantum phase estimation.
The key insight allowing for this
is that important quantum algorithms can be recast in terms of classical signal processing.
Tools from classical signal processing have facilitated the development of
recent quantum algorithms for early fault-tolerant quantum computing~\cite{lin2022heisenberg, zhang2022computing, wang2022quantum}. 
In these algorithms and ours, the quantum computer is used to generate a stochastic estimate of a time signal.
The data from this estimated time signal is then processed with classical computing to return the output of the computation\footnote{As a technical note, such methods vary in the \emph{kernel function} they use to process the time signal estimate. Lin and Tong~\cite{lin2022heisenberg} use a Heaviside function, Wang et al.~\cite{wang2022quantum} use a Gaussian derivative function, and this work uses a Dirichlet kernel function (see Eq. \ref{eq:noiseless_expectation}).}.
Casting this procedure as signal processing, we can answer questions about the performance of this signal processing when the input signal is subject to noise.
This signal processing perspective serves as a foundation for analyzing the robustness of quantum algorithms more generally and a gateway to exploring application to algorithms such as amplitude estimation~\cite{wang2021minimizing}, ground state energy estimation~\cite{dong2022ground}, ground state property estimation~\cite{zhang2022computing}, linear systems solvers~\cite{childs2017quantum}, and factoring~\cite{shor1994algorithms}.

As with several recently developed quantum algorithms~\cite{lin2022heisenberg, huang2020predicting}, the one developed in this work is randomized.
Such algorithms generate random variables, which, in expectation are the desired quantity. 
In the case of our algorithm, the desired quantity is the discrete Fourier transform of the function $g(k)=e^{ik\theta}$ and the randomization is over the ``time variable'' $k$.
As such, we refer to the algorithm as \emph{randomized Fourier estimation (RFE)}.
The algorithm locates the peak of the estimated Fourier spectrum to yield an estimate of $\theta$.
The fact that the algorithm is randomized is what facilitates a simple analysis of its performance in the presence of device error.

\begin{table*}[t]{\footnotesize{
\centering

{\begin{tabular}{|l|c|c}
\cline{1-2}
 &  Expected Runtime  \\ \cline{1-2}
  RFE (Noiseless) Thm. \ref{thm:alg_success} 
  &  $\frac{81\pi^3}{2\epsilon}\ln\left(\frac{8\pi}{\delta\epsilon}\right)$ 

  \\ \cline{1-2}
 RFE (Adversarial Noise) Thm. \ref{thm:BAN_alg_success} 
 & $\frac{81\pi^3}{2\epsilon}\ln\left(\frac{8\pi}{\delta\epsilon}\right){\left(1 - \frac{9\pi}{2\sqrt{2}}\bar{\eta}\right)^{-2}}$  \\ \cline{1-2} 
 RFE (Random Noise) Cor. \ref{thm:GN_alg_success} 
 &  $\frac{81\pi^3}{2\epsilon}\ln\left(\frac{16\pi}{\delta\epsilon}\right) \left( 1-\frac{9\sigma}{8}\sqrt{\frac{\epsilon\pi}{\ln\left(\frac{16\pi}{\delta\epsilon}\right)}}\right)^{-2}$  \\ \cline{1-2}

RPE \cite{PhysRevA.92.062315} & 
$\frac{2\pi}{\epsilon}\log{\frac{1}{\sqrt{\pi}\delta}}\frac{\log{\frac{1}{2}\left(1 - \sqrt{2}\bar{\eta}\right)^{2}}}{\log{1 - \frac{1}{2}\left(1 - \sqrt{2}\bar{\eta}\right)^2}}$

 \\ \cline{1-2} 
\end{tabular}}
\caption{Comparison of the performance of the randomized Fourier estimation (RFE) algorithm in the noiseless setting and under the two noise settings analyzed in this work. In each case, the expected circuit depth (measured in number of $c$-$U$ operations) is $\frac{\pi}{\epsilon}$. The primary feature of the expected runtime in the noisy settings is an additional factor that is greater than 1 and goes to infinity as the noise parameters $\bar{\eta}$ and $\sigma$ reach a certain threshold. For fixed target accuracy $\epsilon$ and success probability $\delta$, as long as the noise parameters are less than a certain value, the algorithm can be guaranteed to succeed with $1-\delta$ probability by taking this factor increase in samples. These algorithms and settings are also compared to the runtime of robust phase estimation (RPE) algorithm of \cite{PhysRevA.92.062315}. For the sake of comparison, we have derived a runtime upper bound for this algorithm in terms of its accuracy and failure rate, and using the fact that $\bar{\eta}$, as defined in Sec. \ref{subsec:ban}, is twice $\delta'$, as defined in Eq. V. 20 of \cite{PhysRevA.92.062315}. While RPE is more performant than RFE for the task of phase estimation, the methodology we use to analyze RFE is generalizable to a broader class of algorithms for early fault-tolerant quantum computers.} 

\label{tab:comparison}}}
\end{table*}

In order to analyze robustness we must model noise.
It is very challenging to accurately model the impact of device error on quantum algorithms.
Accordingly, we will take an intermediate approach.
The effect of device error on a quantum algorithm's performance is mediated through the way that device error affects quantum measurement outcome probabilities.
Rather than propose a model for, say, gate-level quantum circuit errors, we will develop simple, general models for the deviation of outcome probabilities.
We will refer to these as \emph{algorithm error models}. Similar considerations affecting the outcome probabilities have been made in \cite{PhysRevA.92.062315}.
The bounds on the error tolerated by our algorithm are weaker as compared to the bounds on the additive error tolerated by RPE [Theorem 1, \cite{PhysRevA.92.062315}]. However, our algorithm can be extended to developing other robust algorithms for amplitude estimation, ground state energy estimation, etc. That is, by rewriting the input ansatz as a linear combination of eigenstates, the Fourier spectrum generated by our algorithm corresponds to a time signal with multiple peaks which can be detected by RFE by choosing the maximum phase in specified intervals. We leave this analysis for future work. We anticipate that the error will be larger when the initial state is not an exact eigenvector. The error tolerance in such a case is currently being analyzed.

The remaining paper is structured as follows.
In Section \ref{sec:alg_intro} we give a gentle introduction to the phase estimation algorithm and its performance in the noiseless setting.
In Section \ref{sec:robustness} we develop two different algorithm error models and analyze the performance of randomized Fourier estimation under these conditions.
In Section \ref{sec:discussion} we discuss the implications of our analysis and provide an outlook to future work.
For ease of reading, we have reserved the lemmas and proofs of the main theorems to the appendix.

\section{Randomized Fourier estimation}
\label{sec:alg_intro}

The central task we consider is the estimation of the phase angle $\theta$, defined from $$U\ket{\psi}=e^{i\theta}\ket{\psi},$$
where $\ket{\psi}$ is an eigenstate of $U$.
We assume that at our disposal is a preparation of the state $\ket{\psi}$ and the ability to apply the controlled unitary $c-U := \ket{0}\bra{0}\otimes\mathbb{I}+\ket{1}\bra{1}\otimes U$.
Phase estimation is a fundamental subroutine of several quantum algorithms\footnote{An important note is that, in their basic forms, these algorithms require a coherent implementation of quantum phase estimation, whereas this work and others consider incoherent versions of the algorithm.} including Shor’s factoring algorithm \cite{shor1994algorithms}, Harrow, Hassidim, and Lloyd's linear systems algorithm \cite{harrow2009quantum}, and the quantum counting algorithm \cite{brassard1998quantum}.
Furthermore, the robust estimation of quantum phases is closely related to the task of parameter estimation in the context of quantum metrology \cite{giovannetti2004quantum} and Hamiltonian learning \cite{granade2012robust, ferrie2012adaptive}.
Many techniques have been developed in these contexts \cite{svore2013faster, wiebe2016efficient, o2019quantum}
and some have been implemented in experiments \cite{paesani2017experimental, lumino2018experimental}.
Continued research at the interface of quantum computing and quantum metrology will be valuable in pushing the capabilities of quantum computers.

Our goal, slightly different than these above works, is to develop a framework for analyzing the robustness of quantum algorithms for early fault-tolerant quantum computing more broadly.
Towards this, we develop a simple algorithm for the task of phase estimation that shares structural similarity with several other algorithms designed for early fault-tolerant quantum computing including those for ground state energy estimation
\cite{lin2022heisenberg, wang2022quantum} and ground state property estimation (GSPE) \cite{zhang2022computing}. This framework applies to algorithms which use collective data obtained from single-qubit measurement outcomes.
Our hope then is that this framework and methods can be helpful in analyzing the robustness of such quantum algorithms.

Before introducing the algorithm, we note that, after this manuscript originally appeared, several related EFTQC algorithms have been developed and their robustness analyzed with alternative approaches.
Ni et al. \cite{ni2023low} developed a robust quantum algorithm for phase estimation that can succeed even when the input state is imperfect.
Li et al. \cite{li2023adaptive} developed a robust quantum algorithm for multiparameter phase estimation.
A variant of the algorithm introduced below was analyzed in \cite{liang2023modeling}, proving it to be robust to ``exponential decay error'' and showing that the algorithm can reduce quantum resources compared to the traditional quantum phase estimation algorithm.
Most recently, an algorithm for ground state energy estimation \cite{ding2023robust} was developed based on the QCELS method \cite{dong2022ground} and also shown to be robust to exponential decay error.

With this context established, we introduce the robust Fourier estimation algorithm. 
The algorithm generates and processes single-qubit measurement data obtained from Hadamard tests as shown in Figure \ref{fig:hadamard_test}.
\begin{figure}[h!]
\centering
\begin{quantikz}
\lstick{$\ket{0}$} & \gate[wires=1][1cm]{H} & \gate[wires=1][1cm]{S^\ell} & \ctrl{1} & \gate[wires=1][1cm]{H^\dagger} & \meter{} & \\
\lstick{$\ket{\psi}$} & \qwbundle{} & \qw & \gate{U^k} & \qw & \qw &
\end{quantikz}

\caption{Circuit diagram of the Hadamard, where the application of phase gate $S=\begin{bmatrix}
1 & 0\\
0 & i
\end{bmatrix}$ is toggled by $\ell=0,1$.}
\label{fig:hadamard_test}
\end{figure}
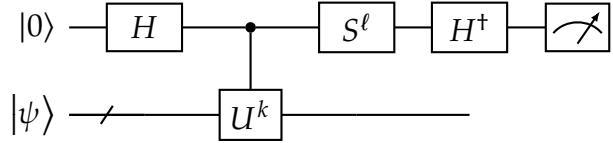

In the case that $\ell=0$, for Hadamard test of $U^k$ the quantum computer gives us $c=\pm1$ outcomes with probability
\begin{align}
    \textup{Pr}(c|k) = \frac{1+c\cos(k\theta)}{2},
\end{align}
where $\textup{Re}(\langle\psi|U|\psi\rangle) = \cos(\theta)=\Pi$.
We call this procedure a \emph{real Hadamard test}.
If a phase gate $S$ is placed before the ancilla measurement (that is $\ell = 1$), the probability of $s=\pm1$ is
\begin{align}
    \textup{Pr}(s|k) = \frac{1+s\sin(k\theta)}{2}.
\end{align}
We call this procedure an \emph{imaginary Hadamard test}.

A simple strategy for estimating $\cos\theta$ (and therefore $\theta$) is to set $k=1$, make a number of Hadamard test measurements with $\ell=0$, and compute the sample mean of $c$.
This approach is often called the \emph{prepare-and-measure} strategy.
Using Hadamard tests with $k\geq 1$, we can gain more information about $\theta$ per measurement compared to the $k=1$ case.
The reason why data from these Hadamard tests gives more information about $\theta$ is that, for large $k$, the above likelihoods are highly sensitive to the value of $\theta$.

We introduce the randomized Fourier estimation algorithm by drawing an analogy between the current setting and a signal processing task.
Consider $k$ to be a "time variable". The biases in the probabilities of the real and imaginary Hadamard tests then encode oscillating time signals,

\begin{align}
\label{eq:sine_cosine_relation_to_probabilities}
   \cos{k\theta} &=\textup{Pr}(c=1|k)-\textup{Pr}(c=-1|k)\\ 
   \sin{k\theta} &=\textup{Pr}(s=-1|k)-\textup{Pr}(s=1|k).
\end{align}
These are the expected values of $c$ and $s$.
If we consider these values to be the real and imaginary parts of a complex time signal $g(k)$, we have
\begin{align}
    g(k) = \cos(k\theta) + i\sin(k\theta) = e^{ik\theta}.
\end{align}
For fixed $k$, if we take a real Hadamard test sample $c$ and an imaginary Hadamard test sample $s$, we can form a random variable whose expected value is the time signal at times $k$,
\begin{align}
    \mathbb{E}(c+is) = e^{ik\theta} = g(k).
\end{align}

Choosing some large maximum time $K$, if we were to collect enough samples to get an accurate estimate of the time signal for times $k=0, \ldots, K-1$, then we might expect to be able to recover the frequency, $\theta$, from this time signal estimate.

One way to recover the frequency is to calculate the discrete Fourier transform (DFT) of the estimated time signal.
In the infinite-sample limit, when the time signal estimate is $g(k)=e^{ik\theta}$, the DFT will peak around a value corresponding to $\theta$ (due to the finite time window, there will be ``leakage'' effects which create additional smaller false peaks (see blue curves in the Fourier spectrum plot of Figure \ref{fig:RFE_example}).
When using a finite, but sufficiently large, number of samples, we expect the DFT of the estimated time signal to still peak around the true peak.
The time signal has $K$ different times. So, getting a constant-accuracy estimate of the time signal requires $\Omega(K)$ samples.
An average sample takes time $\Omega(K)$, leading to a total runtime of $\Omega(K^2)$.
Fortunately, we do not need a very accurate estimate of the time signal in order to get an accurate estimate of the location of the true peak in the Fourier domain.
With just $O(\log{K})$ samples, the peak can be estimated to within $O(\frac{1}{K})$ accuracy with high probability.
The total (expected) runtime is then $O(K\log{K})$.

\begin{figure}[h!]
    \centering
    \includegraphics[width=0.5\textwidth]{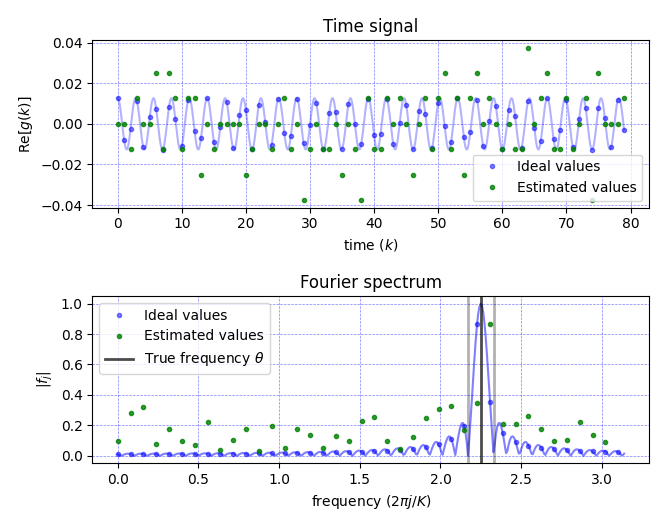}
    \caption{This figure depicts a successful instance of the randomized Fourier estimation algorithm. Despite the fact that the estimate of the time signal (green dots in upper panel) is rather poor, the estimate of the Fourier spectrum (green dots in lower panel) exhibits a spike within $\epsilon=0.08$ (the grey vertical lines) of the true frequency of $\theta=2.25$ (the black vertical line). 
    Here we have chosen to take $M=80$ samples, which is 40 times fewer than the number considered in our algorithm upper bound (c.f. Theorem \ref{thm:alg_success}). The blue dots and blue curves denote the ideal values and continuous functions that define them, respectively.}
    \label{fig:RFE_example}
\end{figure}

A rigorous description of this algorithm is given in Algorithm \ref{alg:dft_rae}.
\begin{algorithm}[H]
\caption{Randomized Fourier Estimation}
\label{alg:dft_rae}
\begin{algorithmic}
 \Require $M$, $K$ $\in \mathbb{N}$ \Comment{Number of samples and grid size, respectively.}
\Ensure $\hat{\theta}\in [0,\pi]$ 
\State $\hat{f} \gets \vec{0}\in\mathbb{C}^K$ \Comment{Initialize estimate of Fourier coefficients.} 
\For{$i \in  \{1, \ldots, M\}$}
\State Draw $k_i$ uniformly from the set $\{0,\ldots,K-1\}$.
\State Call \textsc{Hadamard}$[\ket{\psi},U^{k_i}]$ to output $(c_i,s_i)$.
\For{$j \in \{0, \ldots, K-1\}$}
\State $\hat{f}_j \gets \hat{f}_j+(c_i+is_i)e^{{-2\pi ik_ij}/{K}}/{M}$ \Comment{Update estimate of Fourier coefficient.}
\EndFor
\EndFor
\State $\hat{\theta} \gets 2\pi - \frac{2\pi}{K}\argmax_{j} |\hat{f}_j|$ \Comment{Locate the largest-magnitude Fourier coefficient.}
\end{algorithmic}
\end{algorithm}
We give a simple explanation for how Algorithm \ref{alg:dft_rae} achieves this performance.
The core concept of the algorithm is the fact that, with very few samples, the peak of the discrete Fourier transform can be detected.
Each call to the subroutine \textsc{Hadamard}$[\ket{\psi},U^{k_i}]$ returns the outcomes $c$ and $s$ of the real and imaginary Hadamard tests, respectively, as shown in Figure \ref{fig:hadamard_test}.
Surprisingly, from each real-imaginary Hadamard test sample, we can construct an unbiased estimate of the discrete Fourier transform.
This is achieved with a randomization over the choice of $k\in\{0,\ldots,K-1\}$. Drawing $k$ from the uniform distribution over these values, we form the $j$th discrete Fourier coefficient estimate
$\hat{f}_j=(c+is)e^{\frac{-2\pi i jk}{K}}$, where $j\in \{0, \ldots, K-1\}$. As seen in the following calculation, the expected value of this random variable is the discrete Fourier transform of the time signal $g(k)$,
\begin{equation}
\begin{split}
    \mathbb{E}\hat{f}_j&= \mathbb{E}_{k}\mathbb{E}_{c|k}\mathbb{E}_{s|k}(c+is)e^{\frac{-2\pi i jk}{K}}\\
    &= \frac{1}{K}\sum_{k=0}^{K-1} \exp{ik\theta} e^{\frac{-2\pi i jk}{K}}\\
    &=\frac{1}{K}\frac{1-\exp{iK(\theta-\frac{2\pi j}{K})}}{1-\exp{i(\theta-\frac{2\pi j}{K})}}.
\end{split}
\label{eq:noiseless_expectation}
\end{equation}

We define $f_j := \mathbb{E}\hat{f}_j$.
Viewed as a continuous function in $j$, the magnitude of $f_j$ peaks at $j=\frac{K\theta}{2\pi}$. The algorithm works by averaging $M$ independent Fourier coefficient estimates and then locating the index, or ``frequency'', $j$ for which the magnitude of the estimated coefficient is largest.
This frequency estimate $j$ corresponds to the estimate of $\theta$ via $\hat{\theta}=\frac{2\pi j}{K}$.
We will refer to the $j$ as ``frequencies'', as they index the frequency dimension, while $k$ indexes the time dimension.
The granularity of the estimates is governed by $K$; 
the two frequencies closest to $\theta$ will be off by no more than $\frac{2\pi}{K}$.
Choosing $K=\lceil \frac{2\pi}{\epsilon}\rceil$ then ensures that outputting either of these ``adjacent'' frequencies gives $|\hat{\theta}-\theta|\leq \epsilon$. 
The frequency ``closest'' to $\theta$ (or frequencies closest, in the case they are equidistant from $\theta$) will be off by no more than $\frac{\pi}{K}$.
Our analysis will rely on the fact that this close frequency (or these close frequencies) will have a large expected Fourier coefficient, while the non-adjacent frequencies will have small expected Fourier coefficients\footnote{The algorithm succeeds when either close frequency is the output estimate. Yet, we choose to disregard the properties of the frequency that is (or frequencies that are) adjacent, yet not close (i.e. $j$ satisfying $\frac{1}{2}<|j-\frac{K\theta}{2\pi}|\leq1$). Although the performance bound is looser, in our analysis we only count the algorithm as succeeding when the magnitude of the estimate of the close frequency Fourier component estimate is larger than a fixed value \emph{and} the magnitude of the estimate of all non-adjacent frequency Fourier component estimates are smaller than that fixed value.}.

Why will this estimate be close to $\theta$ with high probability?
The analysis uses two facts about the expected Fourier coefficients: 1) $|f_j|$ of a close frequency is relatively large (Corollary \ref{cor:adj_bound} shows it is greater than $\frac{2}{\pi}$) and 2) for all frequencies $j$ that would give estimates of $\theta$ with error greater than $\frac{2\pi}{K}=\epsilon$, ``non-adjacent frequencies'', $|f_j|$ is relatively small (Corollary \ref{cor:non_adj_bound} shows they are less than $\frac{10}{9\pi}$).
This gap of $\frac{8}{9\pi}$ between $|f_j|$ of close and non-adjacent frequencies gives a buffer for the statistical estimates of $f_j$;
even if all of the estimates $\hat{f}_j$ are off by up to $\frac{4}{9\pi}$, then the magnitude of $\hat{f}_j$ for the close frequency will still be larger than those of the non-adjacent frequencies.
In other words, inaccurate estimates (up to $\frac{4}{9\pi}$ off) still lead the algorithm to succeed.

The final step in the analysis is to assess the likelihood of the close frequency being chosen.
Loosely, the Hoeffding bound ensures that a mean over $M$ samples will deviate by $O(1)$ with no more than $O(\exp{-M})$ chance.
Combining this with a union bound gives that the likelihood that any of $K$ such sample means deviates by more than $O(1)$ is no more than $O(K\exp{-M})$.
Therefore, by setting $K=O(\frac{1}{\epsilon})$, $M=O(\log{\frac{1}{\delta\epsilon}})$ samples will suffice to ensure that the algorithm succeeds (i.e. no non-adjacent frequency is chosen) with greater than $1-\delta$ chance. This is formalized in the following Theorem.
\begin{theorem}[Algorithm success]\label{thm:alg_success}
Given that the inputs are  $M\geq\lceil\frac{81\pi^2}{2}\ln\left(\frac{8\pi}{\delta\epsilon}\right)\rceil$ and $K\geq\lceil \frac{2\pi}{\epsilon}\rceil$, Algorithm \ref{alg:dft_rae} gives an $\epsilon$-accurate estimate of $\theta$ with success probability greater than $1-\delta$.
\end{theorem}
\begin{proof}
If $\epsilon \geq \frac{\pi}{2}$, then $\hat{\theta}=\frac{\pi}{2}$ will be a $\frac{\pi}{2}$-accurate estimate for any $\theta$, so we assume $\epsilon < \frac{\pi}{2}$.
Lemma \ref{lem:inspec_suff} ensures that the algorithm succeeds if $K=\lceil \frac{2\pi}{\epsilon}\rceil > 4$ and all estimated coefficients are in spec (i.e. within $\frac{4}{9\pi}$ of their expected value).
Lemma \ref{lem:inspec_likelihood} ensures that the likelihood of all coefficients being in spec is greater than $1-{4K}\textup{exp} \left(-\frac{2M}{81\pi^2} \right)$.
To ensure that this likelihood is greater than $1-\delta$, we set $M=\lceil\frac{81\pi^2}{2}\ln\left(\frac{8\pi}{\delta\epsilon}\right)\rceil$.
\end{proof}
Since Algorithm \ref{alg:dft_rae} chooses $k$ uniformly from $0$ to $K-1$, the expected runtime of a single circuit is $O(K)$. With this observation, the above theorem implies the following corollary.
\begin{corollary}
Algorithm \ref{alg:dft_rae} can be used to give an $\epsilon$-accurate estimate with success probability greater than $1-\delta$ with expected runtime $O(\frac{1}{\epsilon}\log{\frac{1}{\delta\epsilon}})$.
\end{corollary}

In the following section we investigate the performance of Algorithm \ref{alg:dft_rae} in the presence of noise. We develop several noise models and give a rigorous analysis of the algorithm's success probability.

\section{Provable Robustness of Randomized Fourier Estimation}
\label{sec:robustness}

Here we analyze our algorithm when subject to error in the quantum computation. 
We will show that Algorithm \ref{alg:dft_rae} is robust to a degree of noise under quite general noise models. 
The models will accommodate any type of noise source, including imperfect preparation of $\ket{\psi}$, 
imperfect quantum gate implementation, 
measurement error,
and error introduced by approximating circuit operations (e.g. compilation error or Hamiltonian simulation approximation error).
Whatever the source of these errors, their impact on the algorithm performance is mediated by the way they alter the Hadamard test probabilities:
\begin{align}
\label{eq:BAN_prob_cos}
    \textup{Pr}(c|k) = \frac{1+c( \cos(k\theta)+\eta_{1,k})}{2}
\end{align}
and 
\begin{align}
\label{eq:BAN_prob_sin}
    \textup{Pr}(s|k) = \frac{1+s(\sin(k\theta)+\eta_{2,k})}{2},
\end{align}
where $\eta_{1,k}, \eta_{2,k}$ represent the deviations from the ideal. 
We refer to a model of these alterations as an \emph{algorithm error model}.
In this setting, returning to the time signal analogy as described in Section \ref{sec:alg_intro}, we can view our signal as being corrupted by some error source
\begin{align}
    g(k)\rightarrow\tilde{g}(k) = e^{ik\theta} + \eta_{1,k}+i\eta_{2,k}.
\end{align}
Much like signal recovery in signal processing, we are interested in understanding the degree to which Algorithm \ref{alg:dft_rae} can reconstruct the frequency $\theta$ from this noisy signal.

First, we clarify our approach to algorithm error modeling. We will assume that these deviations, labeled $\eta$, are drawn from some real distribution $\textup{Pr}(\eta)$, which includes the case that the $\eta$ are pre-determined values chosen by an adversary.
Under this model the expected Fourier coefficients will deviate from their noiseless values
\begin{align}
    \mathbb{E}\hat{f}_j &= \mathbb{E}_{\eta}\mathbb{E}_{k}\mathbb{E}_{c|k}\mathbb{E}_{s|k}\hat{f}_j \nonumber\\
    &= \mathbb{E}_{\eta}\mathbb{E}_{k}(e^{ik\theta}+\eta_{1,k}+i\eta_{2,k})e^{-2\pi i jk/K} \nonumber\\
    &= f_j+\eta_j,
\end{align}
where we define $\eta_j:=\mathbb{E}_{\eta}\mathbb{E}_{k}(\eta_{1,k}+i\eta_{2,k})e^{-2\pi i jk/K}$ and  $\hat{\eta}_j:=\mathbb{E}_{k}(\eta_{1,k}+i\eta_{2,k})e^{-2\pi i jk/K}$ for the contribution to the $j$th Fourier coefficient from the drawn noise $\eta$. Moreover, $f_j$ is the expectation of noiseless Fourier coefficient as computed in equation \eqref{eq:noiseless_expectation}.
We emphasize that the impact of noise as modeled above is quite general in that it accommodates any quantum algorithm that processes the outcomes of binary measurement outcomes. This includes algorithms for ground state energy estimation \cite{lin2022heisenberg, wang2022quantum} and ground state property estimation \cite{zhang2022computing}.

In the following two subsections we develop two different noise models (i.e. choices of $\textup{Pr}(\eta)$) and analyze the performance of Algorithm \ref{alg:dft_rae} when subject to them.

\subsection{Bounded adversarial noise model}
\label{subsec:ban}

We first present a worst case analysis under the condition that the noise is bounded. 
 We give two physical settings in which the BAN model can applied to:
\begin{enumerate}
    \item In the first setting (Appendix \ref{appendix:motivation_BAN_model}) we assume Pauli Errors in the state preparation and model the errors occuring in the unitary as the noisy unitary being 
    $\tilde{U} = \mathcal{E} \circ U$ where $\mathcal{E}$ is the Pauli Channel and $U$ is the noiseless unitary. The analogues of equations \eqref{eq:BAN_prob_cos} and \eqref{eq:BAN_prob_sin} are of course long and unwieldy but the general form can be subsumed with in the BAN model  by providing what the $\eta$ is.
    \item The second is a simpler noise model and easier for analysis and is analyzed in this section. It assumes the more dominant source of noise comes for just dephasing errors. Possible realistic settings in which it might be relevant are discussed. 
\end{enumerate}

This setting allows the adversary to choose any deviation of the Hadamard test probabilities as long as they are within some limit.
We can provide performance guarantees as long as this limit is under the bound. Any results we establish that assume bounded adversarial noise will be quite general because we have made limited assumptions about the nature of the noise.
The bounded adversarial noise model assumes the following: once the parameters of the algorithm have been set (e.g. number of samples $M$, number of discrete points $K$, and the true value of the parameter $\theta$), an adversary is able to set $\eta_{1,k}$ and $\eta_{2,k}$ to any value in the interval $[-\bar{\eta},\bar{\eta}]$.
The performance of any algorithm will be dependent on $\bar{\eta}$. For the purpose of our analysis, we assume $\bar{\eta} < \frac{2\sqrt{2}}{9\pi}$. 
While these assumptions allow for a more general setting, in a more realistic noise setting; in practice, noise tends not to act adversarially, but randomly. 
The following theorem summarizes our analysis of the algorithm in the bounded adversarial noise setting.
\begin{theorem}[BAN Model Analysis]\label{thm:BAN_alg_success}
Let the Hadamard tests be subject to the bounded adversarial noise model with bound $\bar{\eta}$. 
Given that the inputs are 
\begin{align}
M\geq\left\lceil\frac{81\pi^2}{2}{\left(1 - \frac{9\pi}{2\sqrt{2}}\bar{\eta}\right)^{-2}}\ln\left(\frac{8\pi}{\delta\epsilon}\right)\right\rceil
\end{align}
and $K\geq\lceil \frac{2\pi}{\epsilon}\rceil$, Algorithm \ref{alg:dft_rae} gives an $\epsilon$-accurate estimate of any $\theta\in[0,\pi]$ with success probability greater than $1-\delta$.
\end{theorem}
\begin{proof}
Lemma \ref{lem:inspec_suff} ensures that the algorithm succeeds if $K=\lceil 2\pi/\epsilon\rceil > 4$ 
and all estimated coefficients are in spec (i.e. within $\frac{4}{9\pi}$ of their expected value).
Lemma \ref{lem:BAN_inspec_likelihood} ensures that, under the bounded adversarial noise model with bound $\bar{\eta}$, Algorithm~\ref{alg:dft_rae} will have all coefficients in spec with likelihood at least $1-{4K}\textup{exp} \left(-\frac{2M\left(1 - \frac{9\pi}{2\sqrt{2}}\bar{\eta}\right)^2}{81\pi^2} \right)$.
To ensure that this likelihood is greater than $1-\delta$, we set $M\geq\lceil\frac{81\pi^2}{2}\left({1} - \frac{9\pi}{2\sqrt{2}}  \bar{\eta}\right)^{-2}\ln\left(\frac{8\pi}{\delta\epsilon}\right)\rceil$ and $K\geq\lceil \frac{2\pi}{\epsilon}\rceil$.
\end{proof}

Here we have chosen a very general noise model. We have shown that the algorithm can still give an accurate estimate despite a degree of error.
To accommodate this error, the algorithm can pay an overhead of $\left({1} - \frac{9\pi}{2\sqrt{2}}  \bar{\eta}\right)^{-2}$ in the number of samples, relative to the noiseless setting.
An interesting feature of this result is that there is an upper threshold $\bar{\eta}<\frac{2\sqrt{2}}{9\pi}\approx 0.1$ above which the algorithm has no success guarantee according to our analysis.
This threshold corresponds to a 0.05 deviation in the Hadamard test probabilities.
Note that when $\bar\eta = 0$, we recover the noiseless case and the bounds of Theorem \ref{thm:BAN_alg_success} match those of Theorem \ref{thm:alg_success}.\\
Next we analyze the physical settings in which the BAN model can be applicable. For a given accuracy, we estimate the required circuit depth as compared to the effective dephasing scale in order for Algorithm \ref{alg:dft_rae} to succeed with high probability.

\subsubsection*{A Physical Setting for the BAN model}
Now we discuss a physical setting in which 
we can apply the BAN model. For this discussion we use the following change in the classical signal introduced in~\cite{wiebe2016efficient}
\begin{align}
    \label{eq:dephasing_noise_signal}
    Pr(c=1|k) = e^{-{k}/{T_2}}\left( \frac{ 1 + \cos (k \theta)}{2} \right) + \frac{1 - e^{-{k}/{T_2}}}{2}.
\end{align}

We consider $k$ as a proxy for the time duration of the quantum circuit in this sense: each application of the unitary is composed of gates which take a certain amount of time to execute. Thus the execution of one application of $U$ takes a certain total amount $T$. Thus the total amount of time taken to execute the whole circuit $k$ times is $kT$. We can see this time scale to 1 and analyze the problem in these "units".
The largest potential deviation from the true probability will occur when $k$ attains its largest value $K$.
Then there are two time scales in this setting: the first is $K$, the time scale taken to run the deepest quantum circuit and the second is $T_2$, the time scale, imposed by the dephasing noise. While this physical setting might be unrealistic for the NISQ setting it's definitely plausible for an early fault tolerant device on platforms like ion-traps where physical (decay) times ($T_1$) are much longer than $T_2$ times, and can be on the order of billions of seconds. Another example where this simple dephasing model may apply is the case where the logical qubits are stored in cat states~\cite{Chamberland2022}. The cat states have been shown to have inherent  exponential suppression of bit-flip errors. Measurements  on the logical qubits to determine the scale of $T_2$ can be done~\cite{Hu2019}. For our discussion, the important quantity is the time scale ratio $\frac{K}{T_2}$. We give an estimate of the time scale ratio for which the algorithm succeeds with high probability. We compute ${\eta}_{1,k}$ as
\begin{equation}
 \label{eq:dephasin_limit_eta_bar}
\tiny{
\begin{split}
\frac{\eta_{1,k}}{2} &= e^{-{k}/{T_2}}\left( \frac{ 1 + \cos (k \theta)}{2} \right) + \frac{1 - e^{-{k}/{T_2}}}{2} - \left( \frac{ 1 + \cos(k\theta)}{2} \right) \\
        &= \frac{e^{-{k}/{T_2}}-1}{2}\cos(k\theta),
\end{split}
}
\end{equation}
and similarly for $\eta_{2,k}$, but by replacing the cosine function with sine function. In order to apply the BAN model, we must determine the largest potential magnitude realized by $\eta_{1,k}$ (and $\eta_{2,k}$) such that
\begin{align}
    |\eta_{1,k}| < \frac{2\sqrt{2}}{9\pi}.
\end{align}
Then we have 
\begin{align}
    \left|\frac{e^{-{k}/{T_2}}-1}{2}\cos(k\theta)\right| < \frac{2\sqrt{2}}{9\pi}.
\end{align}
In the worst case, $|\cos(k\theta)|=1$. Therefore we have
\begin{align}
   0 < \frac{k}{T_2} < -\ln\left(\frac{1}{2}-\frac{2\sqrt{2}}{9\pi}\right) \approx 0.916.
\end{align}
 
This gives $0 < \frac{K}{T_2} <  -\ln\left(\frac{1}{2}-\frac{2\sqrt{2}}{9\pi}\right) $ as the only physical bound on $\frac{K}{T_2}$. This implies that the circuit run-time is allowed to go to at most $0.916$ times the dephasing time before we lose performance guarantee.. \\

In the case that $K << T_2$, Equation \ref{eq:dephasing_noise_signal} can be approximated as 
 \begin{equation}
       \label{eq:high_coherence_limit}
    Pr(c=1|k) = \frac{1+ \cos(k\theta) }{2} + \frac{k}{2T_2},
\end{equation}
leading us to identify ${\eta_{1,k}} = \frac{k}{T_2} $. A similar approximation can be made to obtain $\eta_{2,k} = \frac{k}{T_2}$. 
In this limit, we imagine approximating the noisy likelihood function as getting an error term linearly related to the ratio of circuit depth and coherence time and therefore one interesting question is " what is the maximum depth of the circuit such that this simple model holds and therefore one is allowed to think of about noise entering even more simply?"

More specifically "for a given accuracy $\epsilon$, what is the highest time scale ratio such that simplified Equation \ref{eq:high_coherence_limit} provides performance guarantees"? Requiring that $\epsilon = 0.0004$ with $\bar{\eta} < \frac{2 \sqrt{2}}{9 \pi}$, we get that the dephasing time should be at least 5 times the algorithmic run-time. This means that for the given accuracy  we can approximate  (Equation \ref{eq:dephasing_noise_signal}) as (Equation \ref{eq:high_coherence_limit}) without losing the performance guarantee of the algorithm.

\subsection{Gaussian noise model}

In this subsection we look at a more realistic noise model. We still operate under the assumption that the noise is bounded; however the deviations are chosen from a specific noise model.
The physical setting that we model is that in any one run of the algorithm, the particular setting (e.g. different $\theta$ or different circuit compilations or different devices or different days of using the same device) is drawn randomly.
We model this randomness with Gaussian random variable perturbations of the ideal Hadamard test probabilities.
The \emph{Gaussian noise model} assumes that once the algorithm parameters are set, the Hadamard test biases, $\eta_{1,k}, \eta_{2,k}$, are then drawn independently from Gaussian distributions with standard deviation $\sigma$ and mean $\mu=0$.
When we analyze the success probability of the algorithm, we will assume that in any run of the algorithm the noise parameters are each drawn once and fixed for the remainder of the algorithm.
Accordingly, the expectation and variance of the estimator must include a Gaussian-weighted average over the noise parameters.

Our analysis of the performance of the algorithm under this model will leverage the results from the bounded adversarial noise model.
The approach is as follows.
Observe that for any choice of $\bar{\eta} < \frac{2\sqrt{2}}{9\pi}$, the Gaussian noise model will satisfy the BAN bound with some probability.
In the case where the BAN bound is satisfied, then the probability of success is lower bounded by the results of Theorem \ref{thm:BAN_alg_success}. 
Using the union bound, the probability of failure in the Gaussian noise setting is upper bounded by the sum of the probability of violating the BAN bound and the probability of the algorithm failing given that the BAN bound is satisfied.
We take a strategy that aims to establish a simple expression for the number of measurements required to ensure a success probability greater than $1-\delta$.
This approach compromises the tightness of our results in exchange for simplicity of presentation. Our analysis of the algorithm under the Gaussian noise model is  essentially a corollary of the analysis of the algorithm under the BAN model.

We first choose $\bar{\eta}$ of the BAN model such that the Gaussian noise model violates this bound with no more than $\frac{\delta}{2}$ probability. Lemma \ref{lem:gaussian_failure} gives that we should choose $\bar{\eta}^2=\frac{\sigma^2}{4K}\ln{\left(\frac{8K}{\delta}\right)}$, so that the failure probability upper bound is $\frac{\delta}{2}$.
With this choice of $\bar{\eta}$, the success of the algorithm under the Gaussian noise model is simply a corollary of Theorem \ref{thm:BAN_alg_success}.
\begin{corollary}
[Gaussian Noise Model Analysis]\label{thm:GN_alg_success}
Let the Hadamard tests be subject to the Gaussian noise model with variance $\sigma^2< \frac{64}{81\pi\epsilon}\ln\left(\frac{16\pi}{\delta\epsilon}\right)$. 
Given that the inputs are
\begin{align}
M\geq\left\lceil\frac{81\pi^2}{2}\ln\left(\frac{16\pi}{\delta\epsilon}\right) \left( 1-\frac{9\sigma}{8}\sqrt{\frac{\epsilon\pi}{\ln\left(\frac{16\pi}{\delta\epsilon}\right)}}\right)^{-2}\right\rceil\end{align}
and $K\geq\lceil \frac{2\pi}{\epsilon}\rceil$ 
Algorithm \ref{alg:dft_rae} gives an $\epsilon$-accurate estimate of any $\theta\in[0,\pi]$ with success probability greater than $1-\delta$.
\end{corollary}

This shows that, under a model of random error with bounded variance, with high probability an accurate estimate can be achieved by increasing the number of samples relative to the noiseless case.

\section{Conclusion and Outlook}

\label{sec:discussion}

This work contributes to building a bridge between the noisy intermediate-scale quantum (NISQ) era and the fault-tolerant quantum computing (FTQC) era. We are motivated by the question of how can we predict the 
performance of quantum algorithms in the era when quantum error correction is costly and algorithms will benefit from a degree of robustness. ``Robustness'' is used to describe algorithms with reliable performance even under noisy conditions. In this early fault-tolerant quantum computing (EFTQC) era, there is a need for analytical tools that 
can ensure the robustness of a broad class of quantum algorithms.

\begin{figure*}
    \centering
    \includegraphics[width=0.7\textwidth]{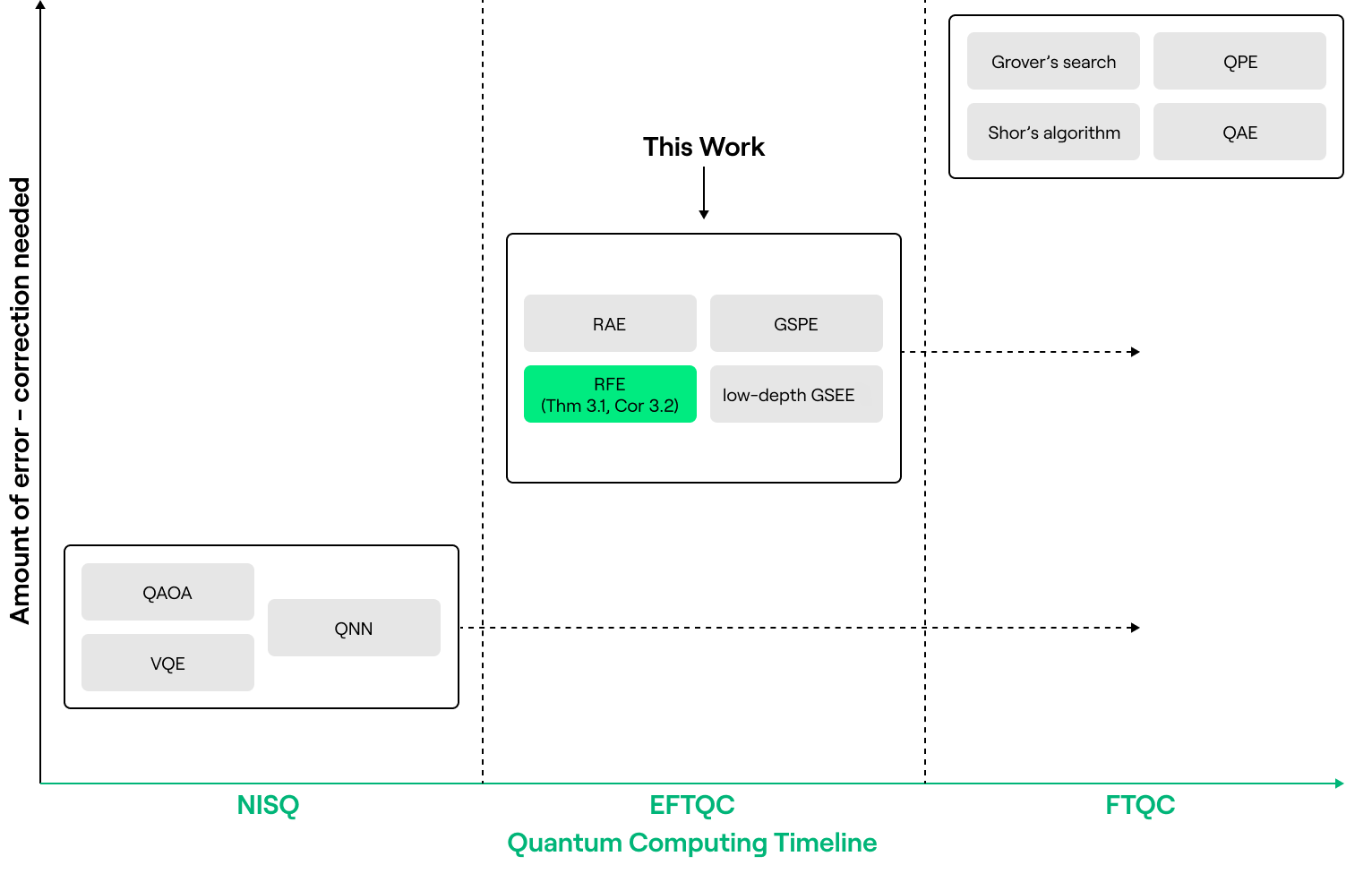}
    \caption{In the EFTQC era, error-correction will be costly and a degree of robustness will benefit quantum algorithms. Theorem \ref{thm:BAN_alg_success} and Corollary \ref{thm:GN_alg_success}  quantify degree of robustness to guarantee success of Algorithm \ref{alg:dft_rae} under noisy conditions. 
    This figure categorizes quantum algorithms according to the degree of quantum error-correction they require and therefore the era of quantum computing for which they are suited. We have used acronyms to abbreviate following algorithms quantum alternating operator ansatz (QAOA)~\cite{farhi2014quantum},  variational quantum eigensolver (VQE)~\cite{peruzzo2014variational}, quantum neural networks (QNN)~\cite{farhi2018classification}, robust amplitude estimation (RAE)~\cite{wang2021minimizing}, low-depth ground state energy estimation (GSEE)\cite{lin2022heisenberg, wang2022quantum}, ground state property estimation (GSPE)\cite{zhang2022computing}, randomized Fourier estimation (RFE, developed in this work), quantum phase estimation (QPE)~\cite{kitaev1995quantum}, and quantum amplitude estimation (QAE)~\cite{brassard2000quantum}.}
    \label{fig:NISQ-FT}
\end{figure*}

To address this need, we develop and analyze a new algorithm for the EFTQC era that exhibits characteristics of many recently-developed quantum algorithms for EFTQC~\cite{lin2022heisenberg, zhang2022computing, wan2022randomized, dong2022ground, wang2022quantum}. As illustrated in Figure \ref{fig:NISQ-FT} the tools that we developed are well-suited for analyzing the robustness of such algorithms. Our analysis paves the way to understanding algorithms that are robust to some degree of error. We leave determining the exact levels of errors and thus determining the cost of error-correction for future work.

The randomized Fourier estimation (RFE) algorithm estimates the phase angle $\theta$ given the operation $c$-$U$ and input state $\ket{\psi}$, where $U\ket{\psi}=e^{i\theta}\ket{\psi}$.
With a single ancilla qubit, using classical randomization and quantum measurements, the algorithm estimates the discrete Fourier spectrum of the time signal $g(k)=e^{i\theta k}$ to learn the location of the spectral peak that encodes $\theta$.
With this prototypical EFTQC quantum algorithm established, we analyze its robustness to error.
We develop a methodology for modeling the impact of error on such algorithms.
To this end, we introduce an important concept: an \emph{algorithmic noise model}, which lets us describe the pertinent features of circuit-level noise on the functioning of the quantum algorithm.
Through this concept we establish rigorous performance guarantees on the RFE algorithm subject to such noise models.
This is a step towards understanding the viability of \emph{EFTQC} and designing algorithms suitable for this era. 

Like many quantum algorithms for phase estimation, our algorithm has a runtime that scales as ${O}(\frac{1}{\epsilon}\log{\frac{1}{\delta\epsilon}})$. 
We expect that the upper bounds obtained in the analysis are not tight and we expect that more sophisticated algorithms for phase estimation~\cite{ji2008parameter, dobvsivcek2007arbitrary, wiebe2016efficient} 
would have better performance.
However, the RFE algorithm facilitates an intuitive analysis of the impact of noise on performance. In the most general noise setting, as long as the noise level is bounded by $\bar{\eta}$, our algorithm has a runtime that scales as ${O}(\frac{1}{\epsilon}\log{\frac{1}{\delta\epsilon}})$ at the cost of extra measurements (dependent on $\bar{\eta}$). We also discuss a physical setting that can be related to such a noise model. In the era of costly quantum error correction it will be important to weigh the cost-benefit of varying degree of quantum error correction. Error correction reduces $\bar{\eta}$ and therefore reduces the expected runtime, but it also comes at its own increased runtime of error correction overhead. Our methodology is a tool to enable researchers to explore such trade-offs. 

Towards such exploration, an important future direction is the analysis of RFE and related algorithms under more realistic noise models. 
In particular, it would be valuable to develop algorithmic noise models which are parameterized by quantities determined from the quantum device. This would help to make the overall methodology more predictive. \cite{Liang:2023dzt} demonstrates a step in this direction.

Our work sets a foundation for future investigations of the robustness of algorithms for EFTQC. Other tasks, such as amplitude estimation, ground state energy estimation, and ground state property estimation, can be viewed as more-involved variants of the phase estimation task.
They vary in the nature of the input state $\ket{\psi}$ and its relationship to the unitary $U$.

For example, ground state energy estimation can be implemented in the setting where $\ket{\psi}$ is a superposition of eigenstates of $H$ and $U$ approximates time evolution under this Hamiltonian~\cite{lin2022heisenberg}.

As with RFE, recent early fault-tolerant algorithms developed for these tasks have a structure that is amenable to our robustness analysis~\cite{lin2022heisenberg, zhang2022computing, wan2022randomized, dong2022ground, wang2022quantum}.
In each of these algorithms, the quantum computer is used to gather time signal data and properties of the signal's spectrum are estimated by postprocessing this data.
Our methodology may be used to model error in time signal and analyze the robustness of such quantum algorithms.
We hope that this work equips researchers with tools to navigate the era of early fault-tolerant quantum computing.

\vspace{1cm}

\noindent\textbf{Acknowledgements} We thank Yiqing Zhou and Qiyao Liang for many insightful discussions during the course of this work. We thank Archismita Dalal, Max Radin, and Daniel Stilck França for helpful feedback on earlier versions of the manuscript and Zofia Włoczewska for help in generating figures.\\

\onecolumn
\appendix

\section{Proofs of Main Results}

 In this section, we discuss the preliminaries necessary to prove results about performance of Algorithm \ref{alg:dft_rae}. 
 The following lemma determines how the distance between the Fourier coefficients estimated in the algorithm and the actual Fourier coefficients impacts the distance between the output phase and actual phase.

\begin{lemma}
\label{lem:inspec_suff}
If $K=\lceil \frac{2\pi}{\epsilon}\rceil\geq 4$ and during Algorithm \ref{alg:dft_rae} $|\hat{f}_j-{f}_j|\leq \frac{4}{9\pi}$ for all $j \in \{0,\ldots,K-1\}$, then Algorithm \ref{alg:dft_rae} outputs $\hat{\theta}$ such that $|\hat{\theta}-\theta|\leq \epsilon$.
\end{lemma}
\begin{proof}
 We will show that when the coefficient estimates are in spec, the algorithm is guaranteed to output one of the two ``adjacent'' discrete frequencies (i.e. those nearest to $\theta$). Choosing $K=\lceil \frac{2\pi}{\epsilon}\rceil$ then ensures that outputting either of these adjacent frequencies gives $|\hat{\theta}-\theta|\leq \epsilon$. 
It remains to show that when all Fourier coefficients are in spec, the algorithm outputs an adjacent frequency. 
Our approach will be to lower bound the magnitude of the expected coefficient of two frequencies adjacent to $\theta$ and upper bound the magnitude of the expected coefficient of any non-adjacent frequency. These bounds will ensure a gap of $\frac{4}{9\pi}$ between the magnitudes of these expected coefficients.
The ``in spec'' distance of $\frac{4}{9\pi}$ then ensures that if all estimated coefficients are within spec, an adjacent frequency will have the largest magnitude of all frequencies.
Consequently, the algorithm will output an adjacent frequency.

The smallest expected magnitude of a ``close'' frequency (i.e. the adjacent frequency closer to $\theta$) is greater than $\frac{2}{\pi}$ and Corollary \ref{cor:non_adj_bound} ensures that the expected magnitude of any non-adjacent frequency is less than $\frac{10}{9\pi}$. 
The minimum distance between the magnitudes of a close frequency and any non-adjacent frequency is greater than $\frac{4}{9\pi}$.
So when the estimates are in spec (i.e. deviating by no more than $\frac{4}{9\pi}$ from their expected value), none of the non-adjacent frequencies can have an estimated coefficient of largest magnitude.
\end{proof}

In order to prove Lemma \ref{lem:SK_upper_bound}, Corollary \ref{cor:adj_bound}, Lemma \ref{lem:SK_lower_bound} and Corollary \ref{cor:non_adj_bound} we will make use of Lemma \ref{lem:coeff_bounds} which gives a general upper and lower bounds on the expected coefficients. %The second is a mathematical statement upper bounding a function within a region.

\begin{lemma}
\label{lem:coeff_bounds}
In Algorithm \ref{alg:dft_rae}, the magnitude of the expected Fourier coefficient of discrete frequency label $j$ evaluates to
\begin{align}
  |{f}_j| = \left|{S_K\left(j-\frac{K\theta}{2\pi}\right)}\right|, 
\end{align}
where $S_{J}(x):=\frac{\sin{\pi x}}{J\sin{(\frac{\pi x}{J})}}$. Note: for large $J$, $S_{J}(x)$ approximates a sinc function, 
$S_J(x)=\frac{\sin(\pi x)}{\pi x}+O(\frac{1}{J^2})$.
% $\displaystyle \lim_{J\rightarrow\infty}S_J(x)=\frac{\sin(\pi x)}{\pi x}$.
\end{lemma}
\begin{proof}
Defining $\omega =\exp{\frac{2\pi i}{K}}$, we can express the expected Fourier coefficient of discrete frequency label $j$ as
\begin{align}
    f_j &= \mathbb{E}_{s,c,k} (c+is)\omega^{jk}\\
    &= \mathbb{E}_{k} e^{ik\theta}\omega^{jk}.
\end{align}
Defining $\Theta = K\theta/2\pi$, we have
\begin{align}
    f_j &= \mathbb{E}_{k} \omega^{-\Theta k}\omega^{jk}\\
    &= \mathbb{E}_{k} \omega^{(-\Theta+j) k}\\\
    &= \frac{1}{K}\left(\frac{1-\omega^{(j-\Theta)K}}{1-\omega^{(j-\Theta)}}\right)\\
    &= \frac{1}{K}\left(\omega^{(j-\Theta)(K-1)/2}\frac{\sin{\pi(j-\Theta)}}{\sin{\frac{\pi (j-\Theta)}{K}}}\right) 
\end{align}
Defining $S_{K}(x)=\frac{\sin{\pi x}}{K\sin{\frac{\pi x}{K}}}$,
we can write the magnitude of the Fourier coefficient as
\begin{align}
    |f_j| = \left|{S_K(j-\Theta)}\right|.
\end{align}
\end{proof}

\begin{lemma}
\label{lem:SK_upper_bound}
For $|x|\leq 1/2$, 
\begin{align}
    \frac{2}{\pi}\leq \left|\frac{\sin{\pi x}}{K\sin{\frac{\pi x}{K}}}\right|.
\end{align}
\end{lemma}

\begin{proof}
Note that $S_K(x)=\frac{\sin{\pi x}}{K\sin{(\frac{\pi x}{K})}}$. On the interval $0 \leq x \leq \frac{1}{2}$ the function $|S_K(x)|$ is a monotonically decreasing function and it achieves minimum at $x = \frac{1}{2}$. Therefore 
\begin{align}
    |S_K(x)|\geq \left|\frac{\sin\left(\frac{\pi}{2}\right)}{K \sin\left(\frac{\pi}{2K}\right)}\right|.
\end{align}
Using that this function is monotonically decreasing with $K$ 
 and for all $K>1$, $|S_K(x)|$ is lower bounded by the limit of this function as $K$ goes to infinity, that is

\begin{align}
|S_K(x)|\geq\frac{2}{\pi}. 
\end{align}
\end{proof}

\begin{corollary}
\label{cor:adj_bound} 
We have for $j$ satisfying $|j-\Theta|\leq \frac{1}{2}$, 
\begin{align}
|f_j|\geq\frac{2}{\pi}.
\end{align}
In other words, the smallest expected magnitude of a close frequency is greater than $\frac{2}{\pi}$. 
\end{corollary}

\begin{proof}
Our starting point is the expression for $|f_j|$ given in Lemma \ref{lem:coeff_bounds},
\begin{align}
|f_j| =\left|{S_K\left(j-\frac{K\theta}{2\pi}\right)}\right| = \left|S_K(j - \Theta)\right|,
\end{align}
where $\Theta = \frac{K\theta}{2\pi}$. Next, the sum $j+\Theta$ must lie between 0 and $K$. Since $j$ corresponds to a close frequency (i.e. $|j-\Theta|\leq \frac{1}{2}$) and because on the interval $0\leq |j-\Theta|\leq \frac{1}{2}$, using Lemma \ref{lem:SK_upper_bound} we have 
\begin{align}
    |f_j|\geq \frac{2}{\pi}.
\end{align}
\end{proof}

\begin{lemma}
\label{lem:SK_lower_bound}
For $1\leq |x|\leq K/2$, $K\geq 4$,
\begin{align}
    \left|\frac{\sin{\pi x}}{K\sin{\frac{\pi x}{K}}}\right|\leq \frac{10}{9\pi}.
\end{align}
\end{lemma}
\begin{proof}
We can use the inequality $S_K(x)=\frac{\sin{\pi x}}{K\sin{(\frac{\pi x}{K})}}\leq \frac{1}{K\sin{(\frac{\pi x}{K})}}$ to write
\begin{align}
|S_K(x)|\leq\frac{1}{\left|K\sin{\left(\frac{\pi x}{K}\right)}\right|}.
\end{align}
Since $\left|\sin\left(\frac{\pi x}{K}\right)\right|$ is monotonically increasing on the interval $1 \leq x \leq \frac{K}{2}$, then we have a minima at $x = 1$. Therefore
\begin{align}
    |S_K(x)| \leq \frac{1}{\left|K\sin(\frac{\pi}{K})\right|}
\end{align}
Furthermore, when $K \geq 4$ we have 
$|S_K(x)|\leq \frac{1}{2\sqrt{2}} \leq \frac{10}{9\pi}$.
\end{proof}

\begin{corollary}
\label{cor:non_adj_bound} 
We have for $j$ satisfying $1 \leq |j-\Theta|\leq \frac{K}{2}$, $K\geq 4$,
\begin{align}
|f_j|\leq\frac{10}{9\pi}.
\end{align}
The expected magnitude of a non-adjacent frequency is less than $\frac{10}{9\pi}$. \\
\end{corollary}
\begin{proof}
As shown in Lemma \ref{lem:coeff_bounds}, the magnitude of the expected value of any coefficient is 
\begin{align}
|{f}_j|= \left|{S_K(j-\Theta)}\right|,
\end{align}
where $\Theta = \frac{K\theta}{2\pi}$. Since both $j$ and $\Theta$ are bound between 0 and $\frac{K}{2}$, so is their difference $j-\Theta$. Furthermore, assuming that $j$ corresponds to a non-adjacent frequency ensures $|j-\Theta|\geq 1$.
For $K\geq 4$, we use Lemma \ref{lem:SK_lower_bound} and have
\begin{align}
|{f}_j|\leq \frac{1}{2\sqrt{2}} \leq \frac{10}{9\pi}.
\end{align}
 \end{proof}

\begin{lemma}[Likelihood of all coefficients being in spec]\label{lem:inspec_likelihood}
In Algorithm \ref{alg:dft_rae}, the likelihood of having $|\hat{f}_j-f_j|\leq \frac{4}{9\pi}$. for all $j \in \{0,\ldots,K-1\}$ is greater than $1-{4K} \textup{exp}\left(-\frac{2M}{81\pi^2} \right)$.
\end{lemma}
\begin{proof}
We use a union bound and a concentration inequality to bound the probability that \emph{all} Fourier coefficients are in spec.

\emph{Union bound:}
The algorithm failure probability is less than the probability that at least one Fourier coefficient is out of spec.
Let $q$ upper bound the probability that the worst-case Fourier coefficient fails to be in spec.
By the union bound, the likelihood that at least one coefficient fails to be in spec is upper bounded as
\begin{align}
\textup{Pr}(\textup{any out of spec}) \leq \sum_l \textup{Pr}(l\textup{th coefficient out of spec}) \leq K q.
\end{align}
Next, we upper-bound $q$.

\emph{Hoeffding bound:}

During Algorithm \ref{alg:dft_rae}, each of the parity samples gives an i.i.d. estimate 
\begin{align*}
    \hat{f}^{(i)}_j&=(c_i + is_i)\exp{\frac{2\pi i k_ij}{K}}\\
    &=(c_i + is_i)\left(\cos\left(\frac{2\pi k_ij}{K}\right) + i\sin\left(\frac{2\pi k_ij}{K}\right)\right),
\end{align*} which are then averaged to obtain $\hat{f}_j=\frac{1}{M}\sum_{i=1}^M \hat{f}^{(i)}_j$.
Thus, we can use a concentration inequality to upper bound the probability that this quantity deviates from its expected value.
The real and imaginary components of a single sample estimate lie within $[-\sqrt{2}, \sqrt{2}]$.
The Hoeffding inequality bounds the likelihood that the sum of $M$ such estimates deviates by more than amount $t$ 
\begin{align}
    \textup{Pr}\left(\left|\sum_{i=1}^M\textup{Re}(\hat{f}_j^{(i)})-\mathbb{E}\sum_{i=1}^M\textup{Re}(\hat{f}_j^{(i)})\right|\geq t\right)\leq 2\textup{exp} \left(-\frac{t^2}{4M} \right),
\end{align}
where we have used that the random variables are bound by an interval of length $2\sqrt{2}$.
Setting $t=\frac{2\sqrt{2}M}{9\pi}$, we can re-express the above statement as
\begin{align}
    \textup{Pr}\left(\left|\frac{1}{M}\sum_{i=1}^M\textup{Re}(\hat{f}_j^{(i)})-\mathbb{E}\textup{Re}(\hat{f}_j)\right|\geq \frac{2\sqrt{2}}{9\pi} \right)\leq 2\textup{exp} \left(-\frac{M}{81\pi^2} \right).
\end{align}

The same relationship holds for the imaginary components.
For any Fourier coefficient, the estimate will be in spec (i.e. distance $\frac{4}{9\pi}$ in the complex plane) if both the real and imaginary components are within $\frac{2\sqrt{2}}{9\pi}$.
Using the union bound, the likelihood that at least one of the components (real or imaginary) deviate by more than $\frac{2\sqrt{2}}{9\pi}$ is less than twice the likelihood of any one deviating by more than this amount.
Putting this together, we can upper bound the likelihood that a given coefficient is out of spec,
\begin{align} 
    &\textup{Pr}\left(j\textup{th coefficient out of spec} \right)=\\
    &\textup{Pr}\left(\left|\hat{f}_j-\mathbb{E}\hat{f}_j\right|\geq \frac{4}{9\pi} \right)\leq  4\textup{exp} \left(-\frac{2M}{81\pi^{2}} \right).
\end{align}
Given that this bound is independent of $j$, we have that $q$, the maximum failure probability over $j$, is bounded as $q\leq 4\textup{exp} \left(-\frac{2M}{81\pi^{2}} \right)$.
Finally, the likelihood that at least one coefficient estimate is out of spec is bounded as
\begin{align}
\textup{Pr}(\textup{any out of spec}) \leq Kq \leq 4K\textup{exp} \left(-\frac{2M}{81\pi^2
} \right).
\end{align}
\end{proof}

\begin{lemma}
\label{lem:BAN_inspec_likelihood}
Subject to bounded adversarial noise with bound $\bar{\eta}< \frac{2\sqrt{2}}{9\pi}$, in Algorithm \ref{alg:dft_rae}, the likelihood of having $|\hat{f}_j-f_j|\leq \frac{4}{9\pi}$ for all $j \in \{0,\ldots,K-1\}$ is greater than $1- 4K\textup{exp} \left(-\frac{2M}{81\pi^2}\left(1 - \frac{9\pi}{2\sqrt{2}}\bar{\eta}\right)^2 \right)$.
\end{lemma}
\begin{proof}
We begin by upper bounding the likelihood that any estimated Fourier coefficient is out of spec.
When including the BAN model, the violation of in-spec is due to a combination of statistical error and BAN error. The algorithm generates noisy estimates of the Fourier coefficient. We use the minimum distance between these noisy estimates and the actual Fourier coefficients to get an upper bound on the probability that  any estimated Fourier coefficient is out of spec, we use the triangular inequality for this purpose. Note that, $\eta_j$ is the deviation in the $j$th iteration of the algorithm and the maximum value of $|\eta_j|$ is $\sqrt{2}\bar{\eta}$, since both $\eta_{1,k}$ and $\eta_{2,k}$ have maximum value $\bar\eta$ and hence the complex vector has maximum length $\sqrt{2}\bar\eta$. 

\begin{align}
    \textup{Pr}\left(\left|\hat{f}_j-f_j\right|\geq \frac{4}{9\pi}\right)&\leq\textup{Pr}\left(\left|\hat{f}_j-f_j-\eta_j\right|\geq \frac{4}{9\pi}-\sqrt{2}\bar{\eta}\right)\\
    &\leq \textup{Pr}\left(\left|\hat{f}_j-\mathbb{E}\hat{f}_j\right|\geq \frac{4}{9\pi} - \sqrt{2}\bar{\eta}\right)
\end{align}

We use a union bound and a concentration inequality to bound the probability that \emph{all} Fourier coefficients are in spec.

\emph{Union bound:}
The algorithm failure probability is less than the probability that at least one Fourier coefficient is out of spec.
Let $q$ upper bound the probability that the worst-case Fourier coefficient fails to be in spec.
By the union bound, the likelihood that at least one coefficient fails to be in spec is upper bounded as
\begin{align}
&\textup{Pr}(\textup{any out of spec, noisy}) \leq\\
&\sum_l \textup{Pr}(l\textup{th coefficient out of spec, noisy}) \leq K q.
\end{align}
Next, we upper-bound $q$.

\emph{Hoeffding bound:}
During Algorithm \ref{alg:dft_rae}, each of the parity samples gives an i.i.d. estimate 
\begin{align*}
  \hat{f}^{(i)}_j&=(c_i + is_i)\exp{\frac{2\pi i k_ij}{K}}\\
  &=(c_i + is_i)\left(\cos\left(\frac{2\pi k_ij}{K}\right) + i\sin\left(\frac{2\pi k_ij}{K}\right)\right),  
\end{align*}
which are then averaged to obtain $\hat{f}_j=\frac{1}{M}\sum_{i=1}^M \hat{f}^{(i)}_j.$
Thus, we can use a concentration inequality to upper bound the probability that this quantity deviates from its expected value.
The real and imaginary components of a single sample estimate lie within $[-\sqrt{2}, \sqrt{2}]$.
The Hoeffding inequality bounds the likelihood that the sum of $M$ such estimates deviates by more than amount $t$
\begin{align}
    \textup{Pr}\left(\left|\sum_{i=1}^M\textup{Re}(\hat{f}_j^{(i)})-\mathbb{E}\sum_{i=1}^M\textup{Re}(\hat{f}_j^{(i)})\right|\geq t\right)\leq 2\textup{exp} \left(-\frac{t^2}{4M} \right),
\end{align}
where we have used that the random variables are bound by an interval of length $2\sqrt{2}$. 
The choice of $t$ depends on $\bar{\eta}$. We choose $t=\frac{M}{\sqrt{2}}\left( \frac{4}{9\pi} - \sqrt{2}\bar{\eta}\right)$ 
according to
our definition of "in-spec" and we can re-express the above statement as
\begin{align}
    &\textup{Pr}\left(\left|\frac{1}{M}\sum_{i=1}^M\textup{Re}(\hat{f}_j^{(i)})-\mathbb{E}\textup{Re}(\hat{f}_j)\right|\geq \left(\frac{4}{9\pi} - \sqrt{2}\bar{\eta}\right)\frac{1}{\sqrt{2}} \right)\\
    &\leq 2\textup{exp} \left(-\frac{2M\left(1 - \frac{9\pi}{2\sqrt{2}}\bar{\eta}\right)^2}{81\pi^2} \right).
\end{align}
The same relationship holds for the imaginary components.
For any Fourier coefficient, the estimate will be in spec (i.e. distance $\frac{4}{9\pi}$ in the complex plane) if both the real and imaginary components are within $\frac{1}{\sqrt{2}}\left(\frac{4}{9\pi} - \sqrt{2}\bar{\eta}\right)$.
Using the union bound, the likelihood that at least one of the components (real or imaginary) deviate by more than $\frac{1}{\sqrt{2}}\left( \frac{4}{9\pi} - \sqrt{2}\bar{\eta}\right)$ is less than twice the likelihood of any one deviating by more than this amount. 

Putting this together, we can upper bound the likelihood that a given coefficient is out of spec,
\begin{align} 
    \textup{Pr}\left(\left|\hat{f}_j-\mathbb{E}\hat{f}_j\right|\geq \frac{4}{9\pi} - \sqrt{2}\bar{\eta}\right)\leq  4\textup{exp}\left(-\frac{2M\left(1 - \frac{9\pi}{2\sqrt{2}}\bar{\eta}\right)^2}{81\pi^2} \right).
\end{align}

The upper bound on the probability that the $j$th estimated Fourier coefficient is out-of-spec depends on $\bar{\eta}$ as
 
\begin{align}
  &\textup{Pr}(j\textup{th coefficient out of spec, noisy})=\\ &\textup{Pr}\left(\left|\hat{f}_j-f_j\right|\geq \frac{4}{9\pi}\right) \\
 & \leq 4\textup{exp} \left(-\frac{2M\left(1 - \frac{9\pi}{2\sqrt{2}}\bar{\eta}\right)^2}{81\pi^2} \right).
\end{align}
Given that this bound is independent of $j$, we have that $q$, the maximum failure probability over $j$, is bounded as $q\leq 4\textup{exp}\left(-\frac{2M\left(1 - \frac{9\pi}{2\sqrt{2}}\bar{\eta}\right)^2}{81\pi^2} \right)$.
Finally, the likelihood that at least one coefficient estimate is out of spec is bounded as
\begin{align}
&\textup{Pr}(\textup{any out of spec, noisy}) \leq Kq \\
&\leq 4K\textup{exp} \left(-\frac{2M\left(1 - \frac{9\pi}{2\sqrt{2}}\bar{\eta}\right)^2}{81\pi^2} \right). 
\end{align}
\end{proof}

\begin{lemma}
\label{lem:complex_variance}
Let $\hat{z}$ be a random complex variable. The complex variance, $\textup{Var}(\hat{z})=\mathbb{E}(|\hat{z}|^2)-|\mathbb{E}(\hat{z})|^2$, allows us to establish the following concentration inequality
\begin{align}
\textup{Pr}\left(|\hat{z}-\mathbb{E}(\hat{z})|\geq t\sqrt{\textup{Var}(\hat{z})}\right) \leq 4\exp{-\frac{t^2}{8}}.
\end{align}
\end{lemma}
\begin{proof}
This inequality is established through the following chain of inequalities:
{\footnotesize \begin{align}
    &\textup{Pr}\left(|\hat{z}-\mathbb{E}(\hat{z})|\geq t\sqrt{\textup{Var}(\hat{z})}\right)  = \textup{Pr}\left(|\hat{z}-\mathbb{E}(\hat{z})|^2\geq t^2\textup{Var}(\hat{z})\right) \label{eq:complex_variance_setup}\\
    \leq & \textup{Pr}\left(|\textup{Re}(\hat{z}-\mathbb{E}(\hat{z}))|\geq \frac{1}{\sqrt{2}}t\sqrt{\textup{Var}(\hat{z})} \textup{ or } |\textup{Im}(\hat{z}-\mathbb{E}(\hat{z}))|\geq \frac{1}{\sqrt{2}}t\sqrt{\textup{Var}(\hat{z})}  \right) \label{eq:complex_variance_complexreal}\\
    \leq & \textup{Pr}\left(|\textup{Re}(\hat{z}-\mathbb{E}(\hat{z}))|\geq \frac{1}{\sqrt{2}}t\sqrt{\textup{Var}(\hat{z})}\right)\nonumber\\
    &+\textup{Pr}\left( |\textup{Im}(\hat{z}-\mathbb{E}(\hat{z}))|\geq \frac{1}{\sqrt{2}}t\sqrt{\textup{Var}(\hat{z})}  \right) \label{eq:complex_variance_unionbound}\\
    \leq & \textup{Pr}\left(|\textup{Re}(\hat{z}-\mathbb{E}(\hat{z}))|
    \geq \frac{1}{\sqrt{2}}t\sqrt{\textup{Var}(\textup{Re}(\hat{z}))}\right)\nonumber\\
    &+\textup{Pr}\left( |\textup{Im}(\hat{z}-\mathbb{E}(\hat{z}))|\geq \frac{1}{\sqrt{2}}t\sqrt{\textup{Var}(\textup{Im}(\hat{z}))}  \right) \label{eq:complex_variance_complexreal_2}\\ 
    \leq & 2\exp{-\frac{t^2}{8}} +  2\exp{-\frac{t^2}{8}} \\
    \leq & 4\exp{-\frac{t^2}{8}}.
\end{align}}
Here we establish an upper bound on the probability of a lower bound on the variance of a complex random variable. This would be useful in the next lemma, Lemma \ref{lem:gaussian_failure}. We start with equation \eqref{eq:complex_variance_setup} as a set up step, equation \eqref{eq:complex_variance_complexreal} is obtained using the definition of absolute value of a complex variable. Equation \eqref{eq:complex_variance_unionbound} is obtained by using the union bound and equation \eqref{eq:complex_variance_complexreal_2} again uses the definition of absolute value of a complex variable. 
\end{proof}

\begin{lemma}
\label{lem:gaussian_failure}
Under the Gaussian noise model with standard deviation $\sigma$ and $\mu=0$, the probability that there is at least one $j$ for which $|\hat{\eta_j}|\geq\bar{\eta}$ is less than $4K\exp{-\frac{\sigma^2}{4K\bar{\eta}^2}}$.
\end{lemma}
\begin{proof}
Using Lemma \ref{lem:complex_variance}, we have $\textup{Pr}(|\hat{\eta_j} - \mathbb{E}_k(\hat{\eta_j})| \geq t\sqrt{\textup{Var}(\hat{\eta_j})}) 
\leq 4\exp{-\frac{t^2}{8}}$.
A union bound over all $j$ gives $\textup{Pr}(\textup{any } |\hat{\eta_j}| \geq t\sqrt{\textup{Var}(\hat{\eta_j})}) \leq 4K\exp{-\frac{t^2}{8}}$.
Setting $t^2=\textup{Var}(\hat{\eta_j})/\bar{\eta}^2$, we find that the failure probability can be upper bounded by evaluating $\textup{Var}(\hat{\eta_j})$. 
\begin{align}
    \textup{Var}(\hat{\eta}_j)&=\mathbb{E}|\hat{\eta}_j|^2\\
    &=\mathbb{E}|\hat{\eta}_{1,j}+i\hat{\eta}_{2,j}|^2\\
    &=2\mathbb{E}|\hat{\eta}_{1,j}|^2\\    
    &=\frac{2}{K^2}\sum_{k=0}^{K-1}\sum_{k'=0}^{K-1}\omega^{j(k-k')}\mathbb{E}\hat{\eta}_{1,k}\hat{\eta}_{1,k'}\\
    &=\frac{2}{K^2}\sum_{k=0}^{K-1}\sigma^2\\     
    &=2\sigma^2/K,
\end{align}
The above equations are obtained using the definition of a complex variable, properties of complex variables, their probabilities and properties of the expectation function. 
This gives
$\textup{Pr}(\textup{any } |\hat{\eta_j}| \geq t\sqrt{\textup{Var}(\hat{\eta_j})}) \leq 4K\exp{-\frac{\sigma^2}{4K\bar{\eta}^2}}$.
\end{proof}
\noindent Note that, in the above calculation, if $\sigma$ is increasing with $k$, which might be a more accurate model of error accumulation in a quantum circuit, we can incorporate this into the calculation of the variance. For example, if $\sigma_k=k\sigma$ for some $\sigma$, then we can replace the final expression with $2K(K-1)(2 K-1)/6\sigma^2/K^2\leq 2K\sigma^2 /3$.\\

\section{Motivation for choice of noise models and relevance to the physical setting}
\label{appendix:motivation_BAN_model}
In this section, we look at a generalized noise model that we use to motivate the BAN model. Svore, Hastings and Freeman (\cite{svore2013faster}) presented the idea of information theoretic phase estimation. They also used ideas from classical signal processing to make their algorithm more efficient.
They consider the circuit in Figure \ref{fig:Faster_phase_estimation_test} for basic measurements.

\begin{figure}[h!]
\centering
\begin{quantikz}
\lstick{$\ket{0}$} &  \gate[wires=1][1cm]{H}  & \gate[wires=1][1cm]{Z(\theta)} &
\ctrl{1}&\gate[wires=1][1cm]{H} & \meter{} & 
\\
\lstick{$\ket{\phi}$} & \qw \qwbundle{}  & \qw & \gate{U^M}  & \qw & &
\end{quantikz}
\caption{(Figure 1 in \cite{svore2013faster})The circuit is used as a measurement operator.} 
\label{fig:Faster_phase_estimation_test}
\end{figure}

\subsection{Preliminary remarks and notation}

For this more general physical setting we shall assume that the initial state is prepared imperfectly and that it's a mixture of pure states i.e,
\begin{equation}
\displaystyle    \rho = \rho^0 \otimes \sum_i a_i \rho_i,
\end{equation}
where $\rho^0$ is the ancilla  state and $\sum_i \rho_i$ is the initial mixed state. One of the terms will be the eigenstate which we denote as $\rho_{\phi}$, where $\phi$ is the corresponding phase. We of course have the restriction that
\begin{equation}
    \sum_i a_i = 1.
\end{equation}

Let us work under the condition that there are errors in the unitary gate, $U$, but crucially we assume the controlled version does not introduce any new noise in the circuit. Then we model the noisy c-U gate as
$$\tilde{U} = \epsilon \circ U,$$
where $\epsilon(\rho) = \sum_\ell \beta_\ell E_\ell \rho E_\ell^\dagger$. Let $\rho^0 = \begin{pmatrix}
1 & 0 \\
0 & 0 \\
\end{pmatrix}$, $\rho^1 = \begin{pmatrix}
0 & 0 \\
0 & 1 \\
\end{pmatrix}$, $\rho^{01} = \begin{pmatrix}
0 & 1 \\
0 & 0 \\
\end{pmatrix}$, $\rho^{10} = \begin{pmatrix}
0 & 0 \\
1 & 0 \\
\end{pmatrix}$, $\rho^{+} = \frac{1}{2}\begin{pmatrix}
1 & 1 \\
1 & 1 \\
\end{pmatrix}$, $\rho^{-} = \frac{1}{2}\begin{pmatrix}
1 & -1 \\
-1 & 1 \\
\end{pmatrix}$, $\rho^{01(H)} = \frac{1}{2}\begin{pmatrix}
1 & -1 \\
1 & -1 \\
\end{pmatrix}$, $\rho^{10(H)} = \frac{1}{2}\begin{pmatrix}
1 & 1 \\
-1 & -1 \\
\end{pmatrix}$. 

\subsection{Calculation for ancilla output} 
Consider the density matrix of the initial state of the noisy input state: 
\begin{equation}
\displaystyle    \rho = \rho^0 \otimes \sum_i a_i \rho_i.
\end{equation}
After applying the first Hadamard gate on the ancilla qubit, the density matrix changes as follows:
\begin{equation}
\displaystyle    \rho = \rho^+ \otimes \sum_i a_i \rho_i.
\end{equation}
After applying the $Z(\theta)$ rotation gate on the ancilla qubit, the density matrix changes as follows:
\begin{equation}
\displaystyle    \rho = \frac{1}{2}[\rho^0 + \rho^1 + \rho^{01}e^{-i\theta} + \rho^{10}e^{i\theta}] \otimes \sum_i a_i \rho_i.
\end{equation}
After applying the controlled rotation gate, $U^M$ with the ancilla qubit as control qubit, the density matrix changes as follows:
\begin{equation}
\begin{split}
\displaystyle    \rho &= [\rho^0 \otimes I + \rho^1 \otimes \tilde{U}^M_{L}]\left(\frac{1}{2}[\rho^0 + \rho^1 + \rho^{01}e^{-i\theta} + \rho^{10}e^{i\theta}] \otimes \sum_i a_i \rho_i\right)[\rho^0 \otimes I + \rho^1 \otimes \tilde{U}^{M^{\dagger}}_{R}]\\
   &= \frac{1}{2}\left[\rho^0 \otimes \sum_i a_i \rho_i + \rho^{10}e^{i\theta}\otimes \tilde{U}^M_{L}\sum_{i} a_i \rho_i + \rho^{01}e^{-i\theta}\otimes \left(\sum_{i} a_i \rho_i \right)\tilde{U}^{M^{\dagger}}_{R} + \rho^1 \otimes   \tilde{U}^M_{L} \left(\sum_{i} a_i\rho_i \right) \tilde{U}^{M^{\dagger}}_{R}\right].
\end{split}
\end{equation}
We use the notation $\tilde{U}^M_{L}$ and $\tilde{U}^{M^{\dagger}}_{R}$ to emphasize that we are considering a right and left multiplication of $U^M$ in $\tilde{U}^M$.

After applying the second Hadamard gate on the ancilla qubit, the density matrix changes as follows:
\begin{equation}
 \rho = \frac{1}{2}\left[\rho^+ \otimes \sum_i a_i \rho_i + \rho^{10(H)}e^{i\theta}\otimes \tilde{U}^M_{L}\sum_{i} a_i \rho_i + \rho^{01(H)}e^{-i\theta}\otimes \left(\sum_{i} a_i \rho_i \right)\tilde{U}^{M^{\dagger}}_{R} + \rho^- \otimes   \tilde{U}^M_{L} \left(\sum_{i} a_i\rho_i \right) \tilde{U}^{M^{\dagger}}_{R}\right].
\end{equation}
This can be rewritten as follows:
\begin{equation}
\begin{split}
 \rho &= \frac{\rho^0 \otimes\left(\sum_i a_i \rho_i + e^{i\theta}\tilde{U}^M_{{L}}\sum_{i} a_i \rho_i + e^{-i\theta}\left(\sum_{i} a_i \rho_i \right)\tilde{U}^{M^{\dagger}}_{{R}} + \tilde{U}^M_{{L}} \left(\sum_{i} a_i\rho_i \right) \tilde{U}^{M^{\dagger}}_{{R}}\right)}{4}\\
&+  \frac{\rho^1 \otimes\left(\sum_i a_i \rho_i - e^{i\theta}\tilde{U}^M_{{L}}\sum_{i} a_i \rho_i - e^{-i\theta}\left(\sum_{i} a_i \rho_i \right)\tilde{U}^{M^{\dagger}}_{R} + \tilde{U}^M_{{L}} \left(\sum_{i} a_i\rho_i \right) \tilde{U}^{M^{\dagger}}_{{R}}\right)}{4}\\
& +\cdots
\end{split}
\end{equation}

\begin{theorem}
{   If $\tilde{U}_L$ = $\epsilon \circ U $ where $ \epsilon(\rho) = \sum_\ell \beta_\ell E_\ell \rho E_\ell^\dagger$  with $\beta_\ell \in \mathbb{R}$, $\sum_{j} \beta_\ell = 1 $, \text{in particular }$\beta_{\ell^*}$ \text{is probability of no error } and $E_\ell \in $ Pauli Group then when $\rho=\rho_\phi$ and $s \in \{3,M\}$
\begin{equation}
\tiny{
    \begin{split}
       \tilde{U}^{s}_L (\rho_\phi) &= e^{2\pi i s \phi_k} \beta_{k_1^*}\beta_{k_2^*}\dots \beta_{k_s^*} \rho_\phi +  e^{2 \pi i s \phi_k} \beta^*_{k_1}\beta_{k_2^*}\dots \beta_{k_{s-1}^*}\sum_{k_s \neq k_s^*} \beta_{k_s} E_{k_s} \rho_\phi E_{k_s}^{\dagger} + \\
       &e^{2\pi i \phi_k (s-1) } \beta_{k_1^*}\beta_{k_2^*}\dots \beta_{k_{s-2}^*}\sum_{k_{s-1} \neq k_{s-1}^*,k_s} \beta_{k_{s-1}} \beta_{k_s} E_{k_s} U E_{k_{s-1}}  \rho_\phi E_{k_{s-1}}^{\dagger} E_{k_s}^{\dagger} + ... \\
       &e^{2\pi i \phi_k (s-j)} \beta_{k_1^*}\beta_{k_2^*}\dots \beta_{k_{s-j}^*}\sum_{k_{s-j} \neq k_{s-j}^*,k_{s-j+1}, \dots k_{s-j}} \beta_{k_{s-j}} \dots \beta_{k_{s}} E_{k_{s}} \dots U E_{k_{s-j}}  \rho_\phi E_{k_{s-j}}^{\dagger} E_{k_{s}}^{\dagger}  + \dots \\
        &e^{2\pi i \phi_k }\sum_{k_1 \neq k_{1}^*,k_2 \dots, k_s} \beta_{k_1} \dots \beta_{k_s} E_{k_s} \dots U E_{k_1}  \rho_{\phi} E_{k_1}^{\dagger} \dots E_{k_s}^{\dagger}  \\
    \end{split}
}
\end{equation}
}
\end{theorem}

\begin{proof}

{Base case: $s=1$, doesn't establish the pattern for an induction argument.}

{
\begin{equation}
\begin{split}
    \tilde{U}_{L}\rho_\phi &= (\epsilon \circ U)\rho_\phi =\epsilon \circ (U\rho_\phi)\\
    &= \epsilon \circ (e^{2\pi i \phi_k}\rho_\phi)\\
    &= e^{2\pi i \phi_k} (\epsilon \circ \rho_\phi)\\
    &= e^{2\pi i \phi_k} \sum_{k_1}\beta_{k_1}E_{k_1}\rho_\phi E_{k_1}^\dagger\\
    &= e^{2\pi i \phi_k} \left(\beta_{k_1^*}\rho_\phi + \sum_{k_1\neq k_1^*}\beta_{k_1}E_{k_1}\rho_\phi E_{k_1}^\dagger\right)\\
\end{split}
\end{equation}
}

{Base case: $s=2$, doesn't establish the pattern for an induction argument.}

{
\begin{equation}
\begin{split}
    \tilde{U}^2_L\rho_\phi &= (\epsilon \circ U)\left(e^{2\pi i \phi_k} \left(\beta_{k_1^*}\rho_\phi + \sum_{k_1\neq k_1^*}\beta_{k_1}E_{k_1}\rho_\phi E_{k_1}^\dagger\right)\right)\\
    &= \epsilon \circ \left(e^{2\pi i \phi_k} \left(\beta_{k_1^*}U\rho_\phi + \sum_{k_1\neq k_1^*}\beta_{k_1}UE_{k_1}\rho_\phi E_{k_1}^\dagger\right)\right)\\
    &= \epsilon \circ \left(e^{4\pi i \phi_k} \beta_{k_1^*}\rho_\phi + e^{2\pi i \phi_k} \sum_{k_1\neq k_1^*}\beta_{k_1}UE_{k_1}\rho_\phi E_{k_1}^\dagger\right)\\
    &= \left(e^{4\pi i \phi_k} \beta_{k_1^*}\sum_{k_2}\beta_{k_2}E_{k_2}\rho_\phi E_{k_2}^\dagger + e^{2\pi i \phi_k} \sum_{k_1\neq k_1^*,k_2}\beta_{k_1}\beta_{k_2}E_{k_2}UE_{k_1}\rho_\phi E_{k_1}^\dagger E_{k_2}^\dagger\right)\\
    &= \left(e^{4\pi i \phi_k} \beta_{k_1^*}\beta_{k_2^*}\rho_\phi + e^{4\pi i \phi_k} \sum_{k_2 \neq k_2^*}\beta_{k_1^*}\beta_{k_2}E_{k_2}\rho_\phi E_{k_2}^\dagger + e^{2\pi i \phi_k} \sum_{k_1\neq k_1^*,k_2}\beta_{k_1}\beta_{k_2}E_{k_2}UE_{k_1}\rho_\phi E_{k_1}^\dagger E_{k_2}^\dagger\right)\\
\end{split}
\end{equation}
}

{Base case: $s=3$}

{
\begin{equation}
\begin{split}
    \tilde{U}^3_L\rho_\phi &= (\epsilon \circ U)\left(e^{4\pi i \phi_k} \beta_{k_1^*}\beta_{k_2^*}\rho_\phi + e^{4\pi i \phi_k} \sum_{k_2 \neq k_2^*}\beta_{k_1^*}\beta_{k_2}E_{k_2}\rho_\phi E_{k_2}^\dagger + e^{2\pi i \phi_k} \sum_{k_1\neq k_1^*,k_2}\beta_{k_1}\beta_{k_2}E_{k_2}UE_{k_1}\rho_\phi E_{k_1}^\dagger E_{k_2}^\dagger\right)\\
    &= \epsilon \circ \left(e^{4\pi i \phi_k} \beta_{k_1^*}\beta_{k_2^*}U\rho_\phi + e^{4\pi i \phi_k} \sum_{k_2 \neq k_2^*}\beta_{k_1^*}\beta_{k_2}UE_{k_2}\rho_\phi E_{k_2}^\dagger + e^{2\pi i \phi_k} \sum_{k_1\neq k_1^*,k_2}\beta_{k_1}\beta_{k_2}UE_{k_2}UE_{k_1}\rho_\phi E_{k_1}^\dagger E_{k_2}^\dagger\right)\\
    &= \epsilon \circ \left(e^{6\pi i \phi_k} \beta_{k_1^*}\beta_{k_2^*}\rho_\phi + e^{4\pi i \phi_k} \sum_{k_2 \neq k_2^*}\beta_{k_1^*}\beta_{k_2}UE_{k_2}\rho_\phi E_{k_2}^\dagger + e^{2\pi i \phi_k} \sum_{k_1\neq k_1^*,k_2}\beta_{k_1}\beta_{k_2}UE_{k_2}UE_{k_1}\rho_\phi E_{k_1}^\dagger E_{k_2}^\dagger\right)\\
    &= e^{6\pi i \phi_k} \beta_{k_1^*}\beta_{k_2^*}\sum_{k_3 \neq k_3^*}E_{k_3}\rho_\phi E_{k_3}^\dagger + e^{4\pi i \phi_k} \sum_{k_2 \neq k_2^*,k_3}\beta_{k_1^*}\beta_{k_2}\beta_{k_3}E_{k_3}UE_{k_2}\rho_\phi E_{k_2}^\dagger E_{k_3}^\dagger\\
    &+ e^{2\pi i \phi_k} \sum_{k_1\neq k_1^*,k_2,k_3}\beta_{k_1}\beta_{k_2}\beta_{k_3}E_{k_3}UE_{k_2}UE_{k_1}\rho_\phi E_{k_1}^\dagger E_{k_2}^\dagger E_{k_3}^\dagger\\
    &= e^{6\pi i \phi_k} \beta_{k_1^*}\beta_{k_2^*}\beta_{k_3^*} \rho_\phi + e^{6\pi i \phi_k} \beta_{k_1^*}\beta_{k_2^*}\sum_{k_3 \neq k_3^*}E_{k_3}\rho_\phi E_{k_3}^\dagger + e^{4\pi i \phi_k} \sum_{k_2 \neq k_2^*,k_3}\beta_{k_1^*}\beta_{k_2}\beta_{k_3}E_{k_3}UE_{k_2}\rho_\phi E_{k_2}^\dagger E_{k_3}^\dagger\\
    &+ e^{2\pi i \phi_k} \sum_{k_1\neq k_1^*,k_2,k_3}\beta_{k_1}\beta_{k_2}\beta_{k_3}E_{k_3}UE_{k_2}UE_{k_1}\rho_\phi E_{k_1}^\dagger E_{k_2}^\dagger E_{k_3}^\dagger\\
\end{split}
\end{equation}
}
To prove the general statement we proceed by strong induction we assume it holds for $s \in\{3,M\} $, then by assumption we have that
\begin{equation}
\tiny{
    \begin{split}
       \tilde{U}^{s-1}_L (\rho_\phi) &= e^{2\pi i (s-1) \phi_k} \beta_{k_1^*}\beta_{k_2^*}\dots \beta_{k_{s-1}^*} \rho_\phi +  e^{2 \pi i (s-1) \phi_k} \beta_{k_1^*}\beta_{k_2^*}\dots \beta_{k_{s-2}^*}\sum_{k_{s-1} \neq k_{s-1}^*} \beta_{k_{s-1}} E_{k_{s-1}} \rho_\phi E_{k_{s-1}}^{\dagger} \\
       &+ e^{2\pi i \phi_k (s-2) } \beta^*_{k_1}\beta^*_{k_2}\dots \beta_{k_{s-3}^*}\sum_{k_{s-2} \neq k_{s-2}^*,k_{s-1}} \beta_{k_{s-2}} \beta_{k_{s-1}} E_{k_{s-1}} U E_{k_{s-2}}  \rho_\phi E_{k_{s-2}}^{\dagger} E_{k_{s-1}}^{\dagger} + \dots \\
       &+ e^{2\pi i \phi_k (s-1-j)} \beta_{k_1^*}\beta_{k_2^*}\dots \beta_{k_{s-j-1}^*}\sum_{k_{s-1-j} \neq k_{s-j-1}^*,k_{s-j}, \dots k_{s-1}} \beta_{k_{s-j-1}} \dots \beta_{k_{s-1}} E_{k_{s-1}} \dots U E_{k_{s-j-1}}  \rho_\phi E_{k_{s-1-j}}^{\dagger} E_{k_{s-1}}^{\dagger}  + \dots \\
        &+ e^{2\pi i \phi_k }\sum_{k_1 \neq k_{1}^*,k_2 \dots, k_{s-1}} \beta_{k_1} \dots \beta_{k_{s-1}} E_{k_{s-1}} \dots U E_{k_1}  \rho_{\phi} E_{k_1}^{\dagger} \dots E_{k_{s-1}}^{\dagger}  \\
    \end{split}
}
\end{equation}
from this we have that
\begin{equation}
\tiny{
 \begin{split}
 (\tilde{U}^{s}_L) \rho_\phi &  =  \tilde{U}_L(\tilde{U}^{s-1}_L \rho_\phi)\\
   &= (\epsilon \circ U) \left(
e^{2\pi i (s-1) \phi_k} \beta_{k_1^*}\beta_{k_2^*}\dots \beta_{k_{s-1}^*} \rho_\phi +  e^{2 \pi i (s-1) \phi_k} \beta_{k_1^*}\beta_{k_2^*}\dots \beta_{k_{s-2}^*}\sum_{k_{s-1} \neq k_{s-1}^*} \beta_{k_{s-1}} E_{k_{s-1}} \rho_\phi E_{k_{s-1}}^{\dagger}  \right.  \\
       &+ e^{2\pi i \phi_k (s-2) } \beta_{k_1^*}\beta_{k_2^*}\dots \beta_{k_{s-3}^*}\sum_{k_{s-2} \neq k_{s-2}^*,k_{s-1}} \beta_{k_{s-2}} \beta_{k_{s-1}} E_{k_{s-1}} U E_{k_{s-2}}  \rho_\phi E_{k_{s-2}}^{\dagger} E_{k_{s-1}}^{\dagger} + \dots \\
      &+ e^{2\pi i \phi_k (s-1-j)} \beta_{k_1^*}\beta_{k_2^*}\dots \beta_{k_{s-j-1}^*}\sum_{k_{s-1-j} \neq k_{s-j-1}^*,k_{s-j}, \dots k_{s-1}} \beta_{k_{s-j-1}} \dots \beta_{k_{s-1}} E_{k_{s-1}} \dots U E_{k_{s-j-1}}  \rho_\phi E_{k_{s-1-j}}^{\dagger} E_{k_{s-1}}^{\dagger}  + \dots \\
       &+  \left. e^{2\pi i \phi_k }\sum_{k_1 \neq k_{1}^*,k_2 \dots, k_{s-1}} \beta_{k_1} \dots \beta_{k_{s-1}} E_{k_{s-1}} \dots U E_{k_1}  \rho_{\phi} E_{k_1}^{\dagger} \dots E_{k_{s-1}}^{\dagger} \right) \\
      &= e^{2 \pi i s \phi_k} \beta_{k_1^*} \dots \beta_{k_{s-1}^*} \textcolor{blue}{ \left( \beta_{k_s^*}\rho_\phi + \sum_{k_{s}\neq k_{s^*}} \beta_{k_s}E_{k_s} \rho_\phi E^{\dagger}_{k_s} \right)}  \\
&+ e^{2 \pi i (s-1) \phi_k} \beta_{k_1^*}\dots \beta_{k_{s-2}^*}\sum_{k_{s-1} \neq k_{s-1}^*,\textcolor{blue}{k_s}} \beta_{k_{s-1}} \textcolor{blue}{\beta_{k_s} E_{k_s}} E_{k_{s-1}} \rho_\phi E_{k_{s-1}}^{\dagger} \textcolor{blue}{E_{k_s}^{\dagger}} + \dots \\
 &+ e^{2 \pi (s-1-j)\phi_k} \beta_{k_1^*} \dots \beta_{k_{s-1-j}} \sum_{k_{s-1-j}\neq k_{s-1-j}^*, k_{s-j},\dots k_{s-1},\textcolor{blue}{k_s} } \beta_{k_{s-1-j}} \dots \beta_{k_{s-1}} \textcolor{blue}{\beta_{k_s} E_{k_s}} U E_{k_{s-1}} \dots U E_{k_{s-1-j}} \rho_\phi E^{\dagger}_{k_{s-1-j}} \dots E^{\dagger}_{k_{s-1}} \textcolor{blue}{E_{k_s}^{\dagger}} \\
 &+ e^{2\pi i \phi_k }\sum_{k_1 \neq k_{1}^*,k_2 \dots, k_{s-1},\textcolor{blue}{k_s}} \beta_{k_1} \dots \beta_{k_{s-1}} E_{k_{s-1}} \textcolor{blue}{\beta_{k_s} E_{k_s}} \dots U E_{k_1}  \rho_{\phi} E_{k_1}^{\dagger} \dots E_{k_{s-1}}^{\dagger} \textcolor{blue}{E^{\dagger}_{k_s}}
 \end{split}
 }
\end{equation}

Hence, proved.
\end{proof}
From the above theorem we have the following expressions
\begin{equation}
\tiny{
    \begin{split}
       \tilde{U}^M _L(\rho_\phi) &= e^{2\pi i M \phi_k} \beta_{k_1^*}\beta_{k_2^*}\dots \beta_{k_M^*} \rho_\phi +  e^{2 \pi i M \phi_k} \beta_{k_1^*}\beta_{k_2^*}\dots \beta_{k_{M-1}^*}\sum_{k_M \neq k_M^*} \beta_{k_M} E_{k_M} \rho E_{k_M}^{\dagger} + \\
       &e^{2\pi i \phi_k (M-1) } \beta^*_{k_1}\beta^*_{k_2}\dots \beta_{k_{M-2}^*}\sum_{k_{M-1} \neq k_{M-1}^*,k_M} \beta_{k_{M-1}} \beta_{k_M} E_{k_M} U E_{k_{M-1}}  \rho E_{k_{M-1}}^{\dagger} E_{k_M}^{\dagger} + \\
       &e^{2\pi i \phi_k (M-n+1)} \beta^*_{k_1}\beta^*_{k_2}\dots \beta_{k_{M-n}^*}\sum_{k_{M-n+1} \neq k_{M-n+1}^*,k_{M-n+2}, \dots k_M} \beta_{k_{M-1}} \beta_{k_M} E_{k_{M-n+1}} \dots U E_{k_{M-n+1}}  \rho E_{k_{M-1}}^{\dagger} E_{k_{M-n+1}}^{\dagger} \dots E_{k_M}^{\dagger} + \\
        &e^{2\pi i \phi_k }\sum_{k_1 \neq k_{1}^*,k_2 \dots, k_M} \beta_{k_1} \dots \beta_{k_M} E_{k_M} \dots U E_{k_1}  \rho E_{k_1}^{\dagger} \dots E_{k_M}^{\dagger}  \\
    \end{split}
}
\end{equation}

\begin{equation}
\tiny{
    \begin{split}
       (\rho_\phi)\tilde{U}^{M^\dagger}_R &= e^{-2\pi i M \phi_k} \beta_{k_1^*}\beta_{k_2^*}\dots \beta_{k_M^*} \rho_\phi +  e^{-2 \pi i M \phi_k} \beta_{k_1^*}\beta_{k_2^*}\dots \beta_{k_{M-1}^*}\sum_{k_M \neq k_M^*} \beta_{k_M} E_{k_M} \rho E_{k_M}^{\dagger} + \\
       &e^{-2\pi i \phi_k (M-1) } \beta_{k_1^*}\beta_{k_2^*}\dots \beta_{k_{M-2^*}}\sum_{k_{M-1} \neq k_{M-1}^*,k_M} \beta_{k_{M-1}} \beta_{k_M} E_{k_{M}} E_{k_{M-1}}  \rho   E_{k_M}^{\dagger} U^{\dagger} E_{k_{M-1}}^{\dagger}  + 
       \\
       &e^{-2\pi i \phi_k (M-n+1)} \beta_{k_1^*}\beta_{k_2^*}\dots \beta_{k_{M-n}^*}\sum_{k_{M-n+1} \neq k_{M-n+1}^*,k_{M-n+2}, \dots k_M} \beta_{k_{M-n_1}} \dots \beta_{k_M} E_{k_M} \dots   E_{k_{M-n+1}} \rho E_{k_{M-n+1}}^{\dagger} \dots U^{\dagger} E_{k_{M}}^{\dagger} + \\
        &e^{-2\pi i \phi_k }\sum_{k_1 \neq k_{1}^*,k_2 \dots, k_M} \beta_{k_1} \dots \beta_{k_M}  E_{k_1} \dots E_{k_M} \rho  E_{k_1}^{\dagger} \dots U^{\dagger} E_{k_M}^{\dagger} + \\
    \end{split}
}
\end{equation}

\begin{equation}
\tiny{
    \begin{split}
      \tilde{U}^{M}_L (\rho_\phi)\tilde{U}^{M^\dagger}_R &=  \beta_{k_1^*}\beta_{k_2^*}\dots \beta_{k_M^*} \rho_\phi +   \beta_{k_1^*}\beta_{k_2^*}\dots \beta_{k_{M-1}^*}\sum_{k_M \neq k_M^*} \beta_{k_M} E_{k_M} \rho E_{k_M}^{\dagger} + \\
       & \beta_{k_1^*}\beta_{k_2^*}\dots \beta_{k_{M-2^*}}\sum_{k_{M-1} \neq k_{M-1}^*,k_M} \beta_{k_{M-1}} \beta_{k_M} E_{k_{M}} UE_{k_{M-1}}  \rho   E_{k_M}^{\dagger} U^{\dagger} E_{k_{M-1}}^{\dagger}  + 
       \\
       & \beta_{k_1^*}\beta_{k_2^*}\dots \beta_{k_{M-n}^*}\sum_{k_{M-n+1} \neq k_{M-n+1}^*,k_{M-n+2}, \dots k_M} \beta_{k_{M-n_1}} \dots \beta_{k_M} E_{k_M} \dots   UE_{k_{M-n+1}} \rho E_{k_{M-n+1}}^{\dagger} \dots U^{\dagger} E_{k_{M}}^{\dagger} + \\
        &\sum_{k_1 \neq k_{1}^*,k_2 \dots, k_M} \beta_{k_1} \dots \beta_{k_M}  E_{k_1} \dots E_{k_M} \rho  E_{k_1}^{\dagger} \dots U^{\dagger} E_{k_M}^{\dagger} \\
    \end{split}
}
\end{equation}

A more straightforward calculation shows that when $\rho$ is not an eigenstate,
\begin{align}
    \tilde{U}^M_{L}\rho &=  \sum_{k_1,k_2,\ldots,k_M} \gamma_{k_1} \gamma_{k_2} \cdots \gamma_{k_M} E_{k_M} \ldots U E_{k_2} U E_{k_1} U \rho E_{k_1}^\dagger E_{k_2}^\dagger \ldots E_{k_M}^\dagger,
\end{align}
\begin{align}
\rho \tilde{U}^{M^{\dagger}}_{R} &=  \sum_{k_1,k_2,\ldots,k_M} \gamma_{k_1} \gamma_{k_2} \cdots \gamma_{k_M} E_{k_M} \ldots E_{k_2} E_{k_1} \rho U^\dagger E_{k_1}^\dagger U^\dagger E_{k_2}^\dagger \ldots U^\dagger E_{k_M}^\dagger,
\end{align}
\begin{align}
\tilde{U}^M_{L} \rho \tilde{U}^{M^{\dagger}}_{R} &= \sum_{k_1,k_2,\ldots,k_M} \gamma_{k_1} \gamma_{k_2} \cdots \gamma_{k_M} E_{k_M}  \ldots U E_{k_2} U E_{k_1} U \rho U^\dagger E_{k_1}^\dagger U^\dagger E_{k_2}^\dagger \ldots U^\dagger E_{k_M}^\dagger.
\end{align}

The probability of measuring $0$ on the ancilla is then

\begin{equation}
\label{eq:prob_0_ancilla_noise}
\tiny{
    \begin{split}
        p(0|x) &= \frac{1}{4} Tr \left( \sum_{j}a_j \rho_j + e^{i\theta} \tilde{U}^M \sum_{j}a_j \rho_j + e^{-i\theta} \sum_{j}a_j \rho_j \tilde{U}^{M^\dagger} + \tilde{U}^M  \sum_{j}a_j \rho_j \tilde{U}^{M^\dagger}  \right) \\
        &= \frac{1}{4} \left( Tr(a_{\phi} \rho_{\phi}) +  \sum_{j \neq \phi} a_j Tr(\rho_j) + a_{\phi} Tr(e^{i\theta} \tilde{U}^M \rho_\phi)  + \sum_{j \neq \phi} a_{j}  Tr (e^{i\theta} \tilde{U}^M \rho_j)  + a_{\phi} Tr(e^{-i\theta}  \rho_\phi \tilde{U}^{M^\dagger}) \right. \\
        & \left. +\sum_{j \neq \phi} a_{j}  Tr (e^{-i\theta} \rho_j\tilde{U}^{M^\dagger}) + a_\phi Tr(\tilde{U}^M \rho_\phi \tilde{U}^{M^\dagger}) + \sum_{j \neq \phi} a_j Tr(\tilde{U}^M \rho_j \tilde{U}^{M^\dagger})
        \right) \\
        &= \frac{1}{4} \left(a_\phi + \sum_{j\neq \phi} a_j + a_\phi + \sum_{j\neq \phi} a_j + 2Re\left[  a_{\phi} + Tr(e^{i\theta \tilde{U}} \rho_\phi) + \sum_{j\neq \phi}a_j r(e^{i\theta \tilde{U}} \rho_i) \right]
        \right) \\
        & = \frac{1}{2} \left( 1 +a_{\phi} \beta_{k_1^*} \dots \beta_{k_{M-1}^*} ( \beta_{k_{M}^*}+ \sum_{k_M \neq k_{M}^* } \beta_{k_M} )\cos(2\pi M \phi_k + \theta ) \right.\\
        &+ a_{\phi} \beta_{k_1^*} \dots \beta_{k_{M-2}^*} \sum_{k_{M-1}\neq k_{M-1}^*,k_M } \beta_{k_{M-1}}  \beta_{k_M} \delta_{k_{M-1},k_M} \cos(2\pi M \phi_k + \alpha_{k_{M-1},k_M} +\theta ) + \dots \\
        & \left. + a_\phi \sum_{k_1\neq k_1^*,k_2, \dots ,k_M } \beta_{k_1} \dots \beta_{k_M} \delta_{k_1, \dots k_M} \cos( 2 \pi \phi_k + \alpha_{k_1, \dots k_M} + \theta ) + \sum_{j\neq \phi}a_j \sum_{k_1, k_2 ,\dots, k_M} \gamma_{k_1} \dots \gamma_{k_m} \sigma_{k_1, \dots k_M} \cos(2 M \phi_k + \eta_{k_1, \dots k_M}+ \theta) \right) \\
    \end{split}
}
\end{equation}
In Equation \eqref{eq:prob_0_ancilla_noise}, $\alpha_i, \eta_i$ represent the additional rotation due to noise, and $\delta_i,\gamma_i,\sigma_i$ represent probabilities ; here $i$ denotes various subscripts used in equation \eqref{eq:prob_0_ancilla_noise}. The number of "noise" terms in Equation \eqref{eq:prob_0_ancilla_noise} as well as the damping of the term corresponding to the required increase as the depth $M$ of the circuit increases. For simplicity, the damping factor as well as the "additive" noise terms in Equation \eqref{eq:prob_0_ancilla_noise} can be consolidated into a single term, $\eta$ representing total noise. This in essence gives us the probability of outcome of the real Hadamard test in Equation \eqref{eq:BAN_prob_cos} which corresponds to the BAN model. Similar calculations can be used to compute the outcome probability of the imaginary Hadamard test as given in Equation \eqref{eq:BAN_prob_sin}.

\bibliographystyle{quantum}

\onecolumn
\appendix

\end{document}